\documentclass[paper=a4, fontsize=11pt]{scrartcl}
\usepackage[T1]{fontenc}
\usepackage{fourier}
\usepackage[english]{babel}								 
\usepackage{amssymb}
\usepackage{amsmath,amsfonts,amsthm} 
\usepackage[pdftex]{graphicx}	
\usepackage{url}
\usepackage[title,titletoc,page]{appendix} 
\usepackage{listings} 
\usepackage{color}
\usepackage{cite}
\usepackage{caption}
\usepackage{subcaption}
\usepackage[ruled,vlined]{algorithm2e}
\usepackage{diagbox}
\usepackage{tabularx}
\usepackage{float}
\usepackage{booktabs}

\numberwithin{equation}{section}		
\numberwithin{figure}{section}			
\numberwithin{table}{section}				

\newtheorem{defi}{Definition}[section]
\newtheorem{lem}{Lemma}[section]
\newtheorem{thm}{Theorem}[section]
\newtheorem{cor}[lem]{Corollary}

\newtheorem{prop}[lem]{Proposition}

\newcommand{\babs}[1]{\Big|{#1}\Big|}

\newcommand{\bonenorm}[1]{\Big|\Big|{#1}\Big|\Big|_1}
\newcommand{\btwonorm}[1]{\Big|\Big|{#1}\Big|\Big|_2}

\newcommand{\vect}[1]{\boldsymbol{\mathbf{#1}}}

\title{Nearly optimal resolution estimate for the two-dimensional super-resolution and a new algorithm for direction of arrival estimation with uniform rectangular array 
 \thanks{\footnotesize This work was supported in part by the Swiss National Science Foundation grant number
			200021--200307.}}
	\author{Ping Liu\thanks{\footnotesize Department of Mathematics, ETH Z\"urich, R\"amistrasse 101, CH-8092 Z\"urich, Switzerland (ping.liu@sam.math.ethz.ch, habib.ammari@math.ethz.ch).}  \and Habib Ammari\footnotemark[2]}

\date{}
\begin{document}
\maketitle

\begin{abstract}
In this paper, we develop a new technique to obtain nearly optimal estimate of the computational resolution limit introduced in \cite{liu2021mathematicaloned, liu2021theorylse, liu2021mathematicalhighd} for two-dimensional super-resolution problems. Our main contributions are fivefold: (i) Our work improves the resolution estimate for number detection and location recovery in two-dimensional super-resolution problems to nearly optimal; (ii)  As a consequence, we derive a stability result for a sparsity-promoting algorithm in two-dimensional super-resolution problems (or Direction of Arrival problems (DOA)). The stability result exhibits the optimal performance of sparsity promoting in solving such problems; (iii) Our techniques pave the way for improving the estimate for resolution limits in higher-dimensional super-resolutions to nearly optimal; (iv)  Inspired by these new techniques, we propose a new coordinate-combination-based model order detection algorithm for two-dimensional DOA estimation and theoretically demonstrate its optimal performance, and (v) we also propose a new coordinate-combination-based MUSIC algorithm for super-resolving sources in two-dimensional DOA estimation. It has excellent performance and enjoys many advantages compared to the conventional DOA algorithms. The coordinate-combination idea seems to be a promising way for multi-dimensional DOA estimation.   
\end{abstract}

\vspace{0.5cm}
\noindent{\textbf{Mathematics Subject Classification:} 94A08,94A12, 42A05, 65J22, 65F99,65K99} 
		
\vspace{0.2cm}
		
\noindent{\textbf{Keywords:} two-dimensional super-resolution, direction of arrival algorithms, resolution estimates, stability results, sparsity-promoting algorithm, model order detection, MUSIC algorithm} 
	\vspace{0.5cm}


\section{Introduction}
It is well-known that the physical nature of wave propagation and diffraction imposes a fundamental barrier in the resolution of imaging systems, which is termed diffraction limit or resolution limit. Since the famous works of Abbe \cite{abbe1873beitrage} and Rayleigh \cite{rayleigh1879xxxi} for quantifying the resolution limit, it is widely used in practice to date that the resolution limit is near half of the wavelength (see, for instance, \cite{book_imaging1,book_imaging2}). Although this kind of resolution limit was widely used, it is lack of mathematical foundations and not that applicable to modern imaging modalities \cite{ ram2006beyond, cohen2019resolution}. From the mathematical perspective, the resolution limit could only be set when taking into account the noise \cite{di1955resolving, den1997resolution, chen2021algorithmic} and surpassing these classical resolution limits is very promising for imaging modalities with high signal-to-noise ratio (SNR). This understanding motivates new works on deriving more rigorous resolution limits \cite{helstrom1964detection, helstrom1969detection, lucy1992statistical, lucy1992resolution}. At the beginning of this century, the dependence of two-point resolution on the noise level has been thoroughly investigated from the perspective of statistical inference \cite{shahram2004imaging, shahram2004statistical, shahram2005resolvability}, but the resolution estimates for resolving multiple sources only achieve breakthroughs in recent years due to its nonlinearity.

To understand the resolution in resolving multiple sources, in the earlier works \cite{liu2021mathematicaloned, liu2021theorylse, liu2021mathematicalhighd} we have defined ``computational resolution limits'' for number detection and location recovery in the one- and multi-dimensional super-resolution problems and characterized them by the signal-to-noise ratio, cutoff frequency, and number of sources. In \cite{liu2021theorylse}, we derived sharp estimates for the computational resolution limits in one dimensional super-resolution problems. We extended the estimations to multi-dimensional cases in \cite{liu2021mathematicalhighd}, but the new estimation is not that sharp due to the techniques of projection used there. Specifically, the upper bound for the resolution increases rapidly as the source number $n$ and space dimensionality $k$ increases. To address this issue, this paper aims to derive better and nearly optimal estimates for the computational resolution limits in two-dimensional super-resolution problems and provide a better way to tackle general multi-dimensional cases. The main contribution of our work are fivefold: (i) Our work improves the resolution estimate in \cite{liu2021mathematicalhighd} for number detection and location recovery in two-dimensional super-resolution problems to nearly optimal; (ii)  As a consequence, we derive a stability result for a sparsity-promoting algorithm in two-dimensional super-resolution problems (or Direction of Arrival problems (DOA)). Although it is well-known that the total variation optimization \cite{candes2014towards} and many other convex optimization based algorithms \cite{tang2014near} have a resolution limit near the Rayleigh limit \cite{tang2015resolution, da2020stable, denoyelle2017support}, our stability result exhibits the optimal super-resolution ability of $l_0$-minimization in solving such problems; (iii) Our techniques reduce the resolution limit problem to a geometric problem, which paves the way for improving the estimate for resolution limits in higher dimensions to nearly optimal; (iv)  Inspired by the techniques used in the proofs, we propose a new coordinate-combination-based model order detection algorithm for two-dimensional DOA problems and demonstrate its optimal performance both theoretically and numerically, and (v) we also propose a new coordinate-combination-based MUSIC (states for MUltiple SIgnal Classification) algorithm for super-resolving sources in two-dimensional DOA estimation. Our original algorithm enjoys certain advantages compared to the conventional DOA algorithms. We also exhibit numerically the phase transition phenomenon of the algorithm, which demonstrates its excellent resolving capacity. The coordinate-combination idea seems to be a promising direction for multi-dimensional DOA estimations.

\subsection{Existing works on the resolution limit problem}
The first theory for quantifying the resolution limit was derived by Ernst Abbe \cite{abbe1873beitrage, volkmann1966ernst}. Since then,  there have been various proposals for the resolution limit \cite{rayleigh1879xxxi, sparrow1916spectroscopic, schuster1904introduction, houston1927compound}, among which the famous and widely used ones are the Rayleigh limit \cite{rayleigh1879xxxi} and the full width at half maximum (FWHM) \cite{demmerle2015assessing}. However, these classical resolution limits neglect the effect of noise and hence are not mathematically rigorous \cite{di1955resolving, den1997resolution, chen2021algorithmic}. From a mathematical perspective, there is no resolution limit when one has perfect access of the exact intensity profile of the diffraction images. Therefore, the resolution limit can only be rigorously set when taking into account the measurement noise or aberration to preclude perfect access to the diffraction images. Based on this understanding, many works were devoted to characterize the dependence of the two-point resolution on the signal-to-noise ratio from the perspective of statistical inference \cite{helstrom1964detection, helstrom1969detection, lucy1992statistical, lucy1992resolution,  shahram2004imaging, shahram2004statistical, shahram2005resolvability}. These classical and semi-classical limits of two-point resolution have been well-studied and we refer the reader to \cite{liu2021mathematical, chen2021algorithmic, de2016limits,  den1997resolution} for more detailed introductions.  

For the resolution limit of superresolving multiple point sources, the problem becomes much more difficult due to the high degree of nonlinearity. To our knowledge, the first breakthrough was achieved by Donoho in 1992 \cite{donoho1992superresolution}. He considered a grid setting where a discrete measure is supported on a lattice (spacing by $\Delta$) and regularized by a so-called  "Rayleigh index" $b$. The problem is to reconstruct the amplitudes of the grid points from their noisy Fourier data in $[-\Omega, \Omega]$ with $\Omega$ being the band limit. He demonstrated that the minimax error for the amplitude reconstruction is bounded from below and above by $SRF^{2b-1}\sigma$ and $SRF^{2b+1}\sigma$ respectively with $\sigma$ being the noise level and the super-resolution factor $SRF = 1/({\Omega \Delta})$. His results emphasize the importance of sparsity and signal-to-noise in super-resolution. But the estimate has not been improved until recent years. In recent years, due to the enormous development of super-resolution modalities in biological imaging \cite{ref4,STED, ref7,PALM,ref9} and the popularity of researches of super-resolution algorithms in applied mathematics \cite{candes2014towards, azais2015spike, duval2015exact, poon2019, tang2013compressed, tang2014near, morgenshtern2016super, morgenshtern2020super, denoyelle2017support, liao2016music, li2020super}, the inherent superresolving capacity of the imaging problem is drawing increasing interest and has been well-studied for the one-dimensional case. In \cite{demanet2015recoverability}, the authors considered resolving $n$-sparse point sources supported on a grid and improved the results of Donoho. They showed that the minimax error in the amplitude recovery scales as $SRF^{2n-1}\sigma$ in the presence of noise with intensity $\sigma$. The case of multi-clustered point sources was considered in \cite{li2021stable, batenkov2020conditioning} and similar minimax error estimations were derived.  In \cite{akinshin2015accuracy, batenkov2019super}, the authors considered the minimax error for recovering off-the-grid point sources. Based on an analysis of the "prony-type system", they derived bounds for both amplitude and location reconstructions of the point sources. More precisely, they showed that  for $\sigma \lessapprox (SRF)^{-2p+1}$, where $p$ is the number of point sources in a cluster, the minimax error for the amplitude and the location recoveries scale respectively as $(SRF)^{2p-1}\sigma$ and $(SRF)^{2p-2} {\sigma}/{\Omega}$, while for the isolated non-clustered source, the corresponding minimax error for the amplitude and the location recoveries scale respectively as $\sigma$ and ${\sigma}/{\Omega}$. We also refer the reader to \cite{moitra2015super,chen2021algorithmic} for understanding the resolution limit from the perceptive of sample complexity and to  \cite{tang2015resolution, da2020stable} for the resolving limit of some algorithms. 

On the other hand, in order to characterize the exact resolution rather than the minimax error in recovering multiple point sources, in the earlier works \cite{liu2021mathematicaloned, liu2021mathematicalhighd, liu2021theorylse} we have defined "computational resolution limits" which characterize the minimum required distance between point sources so that their number and locations can be stably resolved under certain noise level. By developing a nonlinear approximation theory in a so-called Vandermonde space, we have derived sharp bounds for computational resolution limits in the one-dimensional super-resolution problem. In particular, we have showed in \cite{liu2021theorylse} that the computational resolution limits for the number and location recoveries should be respectively $\frac{C_{\mathrm{num}}}{\Omega}(\frac{\sigma}{m_{\min}})^{\frac{1}{2n-2}}$ and  $\frac{C_{\mathrm{supp}}}{\Omega}(\frac{\sigma}{m_{\min}})^{\frac{1}{2n-1}}$, where $C_{\mathrm{num}}$ and  $C_{\mathrm{supp}}$ are constants and $m_{\min}$ is the minimum strength of the point sources. We have extended these estimates to multi-dimensional cases in \cite{liu2021mathematicalhighd} but the results are not that optimal due to the projection techniques used there. In this paper, we improve the estimates for the two-dimensional super-resolution problem by a new technique. The improvements shall be discussed in detail in Section \ref{section:mainresults}. Also, our new technique paves the way for improving the results in higher-dimensional super-resolution problems. 



\subsection{Direction of Arrival estimation}
Our work also inspires new ideas for the two-dimensional direction of arrival estimation. Direction of arrival (DOA) estimation refers to the process of retrieving the direction information of several electromagnetic waves/sources from the received data of a number of antenna elements in a specific array. It is an important problem in array signal processing and finds wide applications in radar, sonar, wireless
communications, etc; see, for instance, \cite{book_imaging2}.

In one-dimensional DOA estimation, if the antenna elements are uniformly spaced in a line, the well-known MUSIC, ESPRIT algorithms, and other subspace methods can resolve the direction of each incident signal/source with high resolution. But for the two-dimensional DOA estimation with regular rectangular array (URA) where both azimuth and elevation angles should be determined, these subspace methods cannot be simply extended to the two-dimensional case to directly determine the azimuth and elevation angle of each source. A major idea to solve the two-dimensional DOA problem is to decompose it into two independent one-dimensional DOA estimations in which the subspaces methods can be leveraged to efficiently restore the direction components of sources corresponding to $x$-axis and $y$-axis. We call the methods with this decoupling idea as one-dimensional-based algorithms throughout the paper for convenience of discussion. It is worth emphasizing that other ways for directly obtaining the azimuth and elevation angles of each source were also considered \cite{zoltowski1996closed, yeh1989estimating, liao2015music}, but the signal processing in a higher dimensional space damped their computational efficiency. 

Although the one-dimensional-based algorithms are usually much more computationally efficient, they still suffer from some issues: (i) the loss of distance separation for $x$-axis or $y$-axis components; (ii) pair matching of the estimated elevation and azimuth angles. For the first issue, the $x$-axis (or $y$-axis) components of two sources may be closely spaced even though the two sources are far away in the two-dimensional space. This causes very unstable reconstruction of the one-dimensional components and the sources. Most of the researches usually ignored these issues and some papers proposed different ways to enhance the reconstruction but the proposed methods are complicated \cite{wang2008tree, wang2015decoupled}. For example, in \cite{wang2008tree}, the authors utilized Taylor expansion, subspace projection, and a tree structure to enhance the reconstruction when the recovered one-dimensional components are unstable. The second issue is that the pair matching of the estimated elevation and azimuth angles is very time consuming when dealing with multiple components of sources. It usually requires a complex process or two-dimensional search \cite{swindlehurst1993azimuth, zoltowski1996closed, del1997matrix, liu1998azimuth, kikuchi2006pair}. 

In this paper, we propose a new efficient  one-dimensional-based algorithm for the two-dimensional DOA estimation which solves the above two issues in a simple way. First, our algorithm employs a new idea named coordinate-combination to avoid severe loss of distance separation between sources in certain region; see Section \ref{section:superioryofalgo} for the detailed discussion. On the other hand, unlike conventional 
one-dimensional-based algorithms, the pair matching problem of our algorithm is a simple balanced assignment problem \cite{pentico2007assignment} which can be solved efficiently by many algorithms such as the Hungarian algorithm.


\subsection{Organization of the paper}
The rest of the paper is organized in the following way. In Section 2, we present the main results on computational resolution limits for the number detection and the location recovery in the two-dimensional super-resolution problem. We also provide a stability result for a sparsity promoting algorithm. In Section 3, we prove the main results in Section 2. Inspired by the techniques in the proofs, in Section 4 and Section 5 we introduce respectively the coordinate-combination-based number detection
and source recovery algorithms in two-dimensional DOA estimations. We also conduct numerical experiments to demonstrate their super-resolution capability. Section 6 presents a nonlinear approximation theory in Vandermonde space which is also a main part in proving our main results. Section 7 is devoted to some conclusions and future works. In the appendix, we prove a technical lemma.

\section{Main results}\label{section:mainresults}
\subsection{Model setting}
We consider the following model of a linear combination of point sources in a two-dimensional space:
\[
\mu=\sum_{j=1}^{n}a_{j}\delta_{\vect y_j},
\]
where $\delta$ denotes Dirac's $\delta$-distribution in $\mathbb R^2$, $\vect y_j \in \mathbb R^2,1\leq j\leq n$, which are the supports of the measure, represent the locations of the point sources and $a_j\in \mathbb C, 1\leq j \leq n$, their amplitudes. We remark that, throughout the paper, we will use bold symbols for vectors and matrices, and ordinary ones for scalar values.
We call that the measure $\mu$ is $n$-sparse if all $a_j$'s are nonzero. 
We denote by
\begin{equation}\label{equ:intendisset}
	m_{\min}=\min_{j=1,\cdots,n}|a_j|,
	\quad 	D_{\min}=\min_{p\neq j}\bonenorm{\vect y_p-\vect y_j}.
\end{equation}
We assume that the available measurement is the noisy Fourier data of $\mu$ in a bounded domain, that is,
\begin{equation}\label{equ:modelsetting1}
	\mathbf Y(\vect{\omega}) = \mathcal F \mu (\vect{\omega}) + \mathbf W(\vect{\omega})= \sum_{j=1}^{n}a_j e^{i \vect{y}_j^\top  \vect{\omega}} + \mathbf W(\vect{\omega}), \ \vect{\omega}\in [0, \Omega]^2,
\end{equation}
where $\mathcal F \mu$ denotes the Fourier transform of $\mu$, 
$\Omega$ is the cut-off frequency, and $\mathbf W$ is the noise. We assume that
\[
||\mathbf W(\vect{\omega})||_\infty< \sigma,
\]
where $\sigma$ is the noise level. We are interested in the resolution limit for a cluster of tightly spaced point sources. 
To be more specific, we denote by
\[
B_{\delta, \infty}(\vect x) := \Big\{ \mathbf y \ \Big|\ \mathbf y\in \mathbb R^2,\  ||\vect y- \vect x||_{\infty}<\delta \Big\},
\] 
and assume that $\vect y_j \in B_{\frac{(n-1)\pi}{6\Omega}, \infty}(\vect 0), j=1,\cdots,n$, or equivalently $||\vect y_j||_{\infty}<\frac{(n-1)\pi}{6\Omega}$.  

\medskip
The inverse problem we are interested in is to recover the discrete measure $\mu$ from the above noisy measurement $\mathbf Y$. 

\subsection{Computational Resolution Limit for number detection in the two-dimensional super-resolution problem}
In this section, we estimate the super-resolving capacity of the source number detection in two-dimensional super-resolution problems. To be specific, we will define and characterize a computational resolution limit for the corresponding number detection problems. Our main results are built upon delicate analysis of the $\sigma$-admissible measure defined below. 

\begin{defi}{\label{def:sigmaadmissiblemeasure}}
	Given a measurement $\mathbf Y$, we say that $\hat \mu=\sum_{j=1}^{m} \hat a_j \delta_{ \mathbf{\hat y}_j}, \ \mathbf{\hat y}_j\in \mathbb R^2$ is a $\sigma$-admissible discrete measure of $\mathbf Y$ if
	\[
	||\mathcal F\hat \mu (\vect{\omega})-\vect Y(\vect{\omega})||_\infty< \sigma, \ \text{for all}\  \vect \omega \in [0, \Omega]^2.
	\]
\end{defi}

Note that the set of $\sigma$-admissible measures of $\mathbf Y$ characterizes all possible solutions to the inverse problem with the given measurement $\mathbf Y$. If all $\sigma$-admissible measures have at least $n$ supports, then detecting the correct source number is possible, for example by targeting at the sparsest admissible measures. However, if there exists one $\sigma$-admissible measure with less than $n$ supports, detecting the source number $n$ is impossible without additional prior information. This leads to the following new definition of resolution limit, named computational resolution limit.

\begin{defi}\label{def:computresolutionlimitnumber}
The computational resolution limit to the number detection problem in two dimensions  is defined as the smallest nonnegative number $\mathcal D_{2,num}$ such that for all $n$-sparse measures $\sum_{j=1}^{n}a_{j}\delta_{\mathbf y_j}, \vect y_j \in B_{\frac{(n-1)\pi}{6\Omega}, \infty}(\vect 0)$ and the associated measurement $\vect Y$ in (\ref{equ:modelsetting1}), if 
	\[
	\min_{p\neq j} ||\mathbf y_j-\mathbf y_p||_1 \geq \mathcal D_{2, num},
	\]
then there does not exist any $\sigma$-admissible measure  with less than $n$ supports for $\mathbf Y$.
\end{defi}

The above resolution limit is termed ``computational resolution limit'' to  distinguish it from the classic Rayleigh limit. Compared to the Rayleigh limit, the definition of the computational resolution limit is more rigorous from the mathematical perspective. It is related to the noise, by which it is more applicable for modern imaging techniques. In \cite{liu2021mathematicaloned, liu2021theorylse, liu2021mathematicalhighd}, the authors defined similar computational resolution limits and present rigorous estimations for them. Here by the following theorem, we derive a nearly optimal estimate to the $\mathcal D_{2, num}$, which substantially improves the estimate in \cite{liu2021mathematicalhighd} for the two-dimensional case. 

\begin{thm}\label{thm:highdupperboundnumberlimit0}
	Let the measurement $\mathbf Y$ in (\ref{equ:modelsetting1}) be generated by a $n$-sparse measure $\mu =\sum_{j=1}^{n}a_j\delta_{\mathbf y_j}, \vect y_j \in B_{\frac{(n-1)\pi}{6\Omega}, \infty}(\vect 0)$. Let $n\geq 2$ and assume that the following separation condition is satisfied 
	\begin{equation}\label{equ:highdupperboundnumberlimit1}
		\min_{p\neq j, 1\leq p, j\leq n}\Big|\Big|\mathbf y_p- \mathbf y_j\Big|\Big|_1\geq \frac{16.6\pi (n-1)}{\Omega }\Big(\frac{\sigma}{m_{\min}}\Big)^{\frac{1}{2n-2}}.
	\end{equation}
	Then there does not exist any $\sigma$-admissible measures of \,$\mathbf Y$ with less than $n$ supports.
\end{thm}

Theorem \ref{thm:highdupperboundnumberlimit0} reveals that when $\min_{p\neq j, 1\leq p, j\leq n}\Big|\Big|\mathbf y_p- \mathbf y_j\Big|\Big|_1\geq \frac{16.6\pi (n-1)}{\Omega }\Big(\frac{\sigma}{m_{\min}}\Big)^{\frac{1}{2n-2}}$, recovering exactly the source number $n$ is possible. Compared with the Rayleigh limit $\frac{c_2\pi}{\Omega}$, where $c_2$ is a constant, Theorem \ref{thm:highdupperboundnumberlimit0} also indicates that resolving the source number in the sub-Rayleigh regime is theoretically possible if the SNR is sufficiently large. 

Moreover, the estimate in Theorem \ref{thm:highdupperboundnumberlimit0} substantially improves the result in \cite{liu2021mathematicalhighd}, where the upper bound estimation for the two-dimensional computational resolution limit is 
\[
\frac{Cn(n-1)}{\Omega}\Big(\frac{\sigma}{m_{\min}}\Big)^{\frac{1}{2n-2}}
\]
with $C$ being an explicit constant. By the techniques of this paper, we also pave the way for estimating the resolution limit for higher dimensions. It is indicated that we can demonstrate that the corresponding resolution limit in the $k$-dimensional super-resolution problem can be bounded above by 
\[
\frac{C_{num}(k)(n-1)}{\Omega}\Big(\frac{\sigma}{m_{\min}}\Big)^{\frac{1}{2n-2}},
\] 
where $C_{num}(k)$ is a constant determined by the space dimensionality. This  substantially improves the result in \cite{liu2021mathematicalhighd} that the computational resolution limit is estimated to be bounded above by 
\[
\frac{4.4\pi e \ (\pi/2)^{k-1} (n(n-1)/\pi)^{\xi(k-1)}}{\Omega }\Big(\frac{\sigma}{m_{\min}}\Big)^{\frac{1}{2n-2}},
\]
where $\xi(k)= \sum_{j=1}^{k}\frac{1}{j}, \ k\geq 1$. By these new estimates, we get rid of the exponential dependence of the index  $n$ on the dimensionality $k$.

On the other hand, it is already known from \cite{liu2021mathematicalhighd} that the computational resolution limit for the number detection in the $k$-dimensional super-resolution  problem is bounded below by $\frac{C_{1}}{\Omega}\Big(\frac{\sigma}{m_{\min}}\Big)^{\frac{1}{2n-2}}$ for some constant $C_1$. Thus the $\mathcal D_{2,num}$ is bounded by 
\begin{equation}\label{equ:twodnumberequ1}
\frac{C_{1}}{\Omega}\Big(\frac{\sigma}{m_{\min}}\Big)^{\frac{1}{2n-2}} \leq \mathcal D_{2,num} \leq \frac{C_{2}n}{\Omega}\Big(\frac{\sigma}{m_{\min}}\Big)^{\frac{1}{2n-2}}. 
\end{equation}
This estimate is nearly optimal. 

The above estimates further indicate a phase transition phenomenon in the two-dimensional number detection problem. Specifically, by (\ref{equ:twodnumberequ1}) we expect the presence of a line of slope $2n-2$ in the parameter space $\log(SRF)-\log(SNR)$ above which the source number can be correctly detected in each realization. This phenomenon is confirmed exactly by the number detection algorithm (\textbf{Algorithm \ref{algo:coordcombinsweepnumberalgo}}) later in Section \ref{section:numberphasetransition} and illustrated in Figure \ref{fig:twodnumberphasetransition}. 


\subsection{Computational Resolution Limit for location recovery in the two-dimensional super-resolution problem}
We next present our results on the resolution limit for the location recovery problem in two-dimensions. We first introduce the following concept of $\delta$-neighborhood of discrete measures. \\
Define 
\[
B_{\delta, 1}(\vect x) := \Big\{ \mathbf y \ \Big|\ \mathbf y\in \mathbb R^2,\  ||\vect y- \vect x||_{1}<\delta \Big\}. 
\]

\begin{defi}\label{deltaneighborhood}
	Let  $\mu=\sum_{j=1}^{n}a_j \delta_{\vect y_j}$ be a $n$-sparse discrete measure in $\mathbb R^2$ and let $\delta>0$ be such that the $n$ balls $B_{\delta, 1}(\vect y_j), 1\leq j \leq n$ are pairwise disjoint. We say that 
	$\hat \mu=\sum_{j=1}^{n}\hat a_{j}\delta_{\mathbf{\hat y}_j}$ is within $\delta$-neighborhood of $\mu$ if each $\mathbf {\hat y}_j$ is contained in one and only one of the $n$ balls $B_{\delta, 1}(\vect y_j), 1\leq j \leq n$.
\end{defi}

According to the above definition, a measure $\hat \mu$ in a $\delta$-neighborhood of $\mu$ preserves the inner structure of the collection of point sources. For a stable location (or support of measure) recovery algorithm, the output should be a measure in some $\delta$-neighborhood of the underlying sources. Moreover, $\delta$ should tend to zero as the noise level $\sigma$ tends to zero. We now introduce the computational resolution limit for the support recovery problem. For ease of exposition, we only consider measures supported in $B_{\frac{(2n-1)\pi}{12\Omega}, \infty}(\vect 0)$, where $n$ is the source number. 

\begin{defi}\label{def:computresolutionlimitsupport}
	The computational resolution limit in the two-dimensional location recovery problem is defined as the smallest non-negative number $\mathcal D_{2,supp}$ so that for
	any $n$-sparse measure $\mu=\sum_{j=1}^{n}a_j \delta_{\mathbf y_j}, \vect y_j \in B_{\frac{(2n-1)\pi}{12\Omega}, \infty}^{k}(\vect 0)$ and the associated measurement $\vect Y$ in (\ref{equ:modelsetting1}), 
	if
	\[
	\min_{p\neq j, 1\leq p,j \leq n} \bonenorm{\mathbf y_p-\mathbf y_j}\geq \mathcal{D}_{2,supp},
	\]  
	then there exists $\delta>0$ such that any $\sigma$-admissible measure of $\mathbf Y$ with $n$ supports in $B_{\frac{(2n-1)\pi}{12\Omega}}(\mathbf 0)$ is within $\delta$-neighbourhood of $\mu$.  
\end{defi}

We have the following estimate for the upper bound of $\mathcal{D}_{2,supp}$. 
\begin{thm}\label{thm:highdupperboundsupportlimit0}
	Let $n\geq 2$. Let the measurement $\vect Y$ in (\ref{equ:modelsetting1}) be generated by a $n$-sparse measure $\mu=\sum_{j=1}^{n}a_j \delta_{\vect y_j}, \vect y_j \in B_{\frac{(2n-1)\pi}{12\Omega}, \infty}(\vect 0)$ in the two-dimensional space. Assume that
	\begin{equation}\label{equ:highdsupportlimithm0equ0}
		D_{\min}:=\min_{p\neq j}\Big|\Big|\mathbf y_p-\mathbf y_j\Big|\Big|_1\geq \frac{15.3\pi(n-\frac{1}{2})}{\Omega }\Big(\frac{\sigma}{m_{\min}}\Big)^{\frac{1}{2n-1}}. \end{equation}
If $\hat \mu=\sum_{j=1}^{n}\hat a_{j}\delta_{\mathbf{\hat y}_j}$ supported on $B_{\frac{(2n-1)\pi}{12\Omega}, \infty}(\vect 0)$ is a $\sigma$-admissible measure of $\vect Y$, then $\hat \mu$ is in a $\frac{D_{\min}}{2}$-neighborhood of $\mu$. Moreover, after reordering the $\mathbf{\hat y}_j$'s, we have 
	\begin{equation}\label{equ:highdsupportlimithm0equ1}
	\Big|\Big|\mathbf {\hat y}_j-\mathbf y_j\Big|\Big|_1\leq \frac{C(n)}{\Omega}SRF^{2n-2}\frac{\sigma}{m_{\min}}, \quad 1\leq j\leq n,
	\end{equation}
	where $SRF:=\frac{\pi}{D_{\min}\Omega} $ is the super-resolution factor and $$C(n)=\frac{(1+\sqrt{3})^{2n-1}2^{5n-1}(2n-1)^{2n-1}\pi}{3^{2n-0.5}}.$$
\end{thm}

Theorem \ref{thm:highdupperboundnumberlimit0} demonstrates that when $\min_{p\neq j, 1\leq p, j\leq n}\Big|\Big|\mathbf y_p- \mathbf y_j\Big|\Big|_1\geq \frac{15.3\pi (n-1)}{\Omega }\Big(\frac{\sigma}{m_{\min}}\Big)^{\frac{1}{2n-2}}$, it is possible to recover stably the source locations. For sufficiently large SNR, the limit in  Theorem \ref{thm:highdupperboundnumberlimit0} is less than the Rayleigh limit. This indicates that super-resolution is possible for two-dimensional imaging problems. Also, the estimate here is better than the one obtained in \cite{liu2021mathematicalhighd}, which is  
\[
\frac{Cn(n-1)}{\Omega}\Big(\frac{\sigma}{m_{\min}}\Big)^{\frac{1}{2n-1}},
\]
with an explicit constant $C$. By the techniques of this paper, we also pave the way for estimating the resolution limit of location recovery in higher dimensional super-resolution problems. In fact,  the corresponding resolution limit in the $k$-dimensional super-resolution problem can be bounded above by 
\[
\frac{C_{supp}(k)(n-1)}{\Omega}\Big(\frac{\sigma}{m_{\min}}\Big)^{\frac{1}{2n-1}},
\] 
where $C_{supp}(k)$ is a constant depending on the space dimensionality. This substantially improves the estimate, 
\[
\frac{5.88\pi e 4^{k-1}((n+2)(n-1)/2)^{\xi(k-1)}}{\Omega }\Big(\frac{\sigma}{m_{\min}}\Big)^{\frac{1}{2n-1}}, \quad \xi(k)= \sum_{j=1}^{k}\frac{1}{j}, \ k\geq 1, 
\]
derived in \cite{liu2021mathematicalhighd}. Again, one can get rid of the exponential dependence of the index of $n$ on the dimensionality $k$ by using these new estimates.  

It has been already shown in \cite{liu2021mathematicalhighd} that the computational resolution limit for  the location recovery in the $k$-dimensional super-resolution problem is bounded below by $\frac{C_{3}}{\Omega}\Big(\frac{\sigma}{m_{\min}}\Big)^{\frac{1}{2n-1}}$ for some constant $C_3$. Thus the $\mathcal D_{2,supp}$ is bounded by 
\begin{equation}\label{equ:twodsupportequ1}
\frac{C_{3}}{\Omega}\Big(\frac{\sigma}{m_{\min}}\Big)^{\frac{1}{2n-1}} \leq \mathcal D_{2,supp} \leq \frac{C_{4}n}{\Omega}\Big(\frac{\sigma}{m_{\min}}\Big)^{\frac{1}{2n-1}}. 
\end{equation}
This estimate is nearly optimal. 

On the other hand, (\ref{equ:twodsupportequ1}) indicates a phase transition in the location recovery problem. From (\ref{equ:twodsupportequ1}) we expect that there exists a line of slope $2n-1$ in the parameter space of $\log SRF-\log SNR$ such that the location recovery is stable in every point above the line. This is confirmed by \textbf{Algorithm \ref{algo:coordcombinMUSICalgo}} in Section \ref{section:locationphasetransition} and illustrated in Figure \ref{fig:twodsupportphasetransition}.


\subsection{Stability of a sparsity-promoting algorithm}
Sparsity-promoting algorithms are popular methods in imaging processing and many other fields. By the results for resolution limit, we can derive a stability result for a $l_0$-minimization in the two-dimensional super-resolution problems. We consider the following $l_0$-minimization problem:
\begin{equation}\label{equ:l0minimization}
\min_{\rho\in \mathcal O} ||\rho||_{0} \quad \text{subject to} \quad |\mathcal F\rho(\vect \omega) -\vect Y(\omega)|< \sigma, \quad \vect \omega \in[0,\Omega]^2, 
\end{equation}	
where $||\rho||_{0}$ is the number of Dirac masses representing the discrete measure $\rho$. As a corollary of Theorems \ref{thm:highdupperboundnumberlimit0} and \ref{thm:highdupperboundsupportlimit0}, we have the following stability result. 

\begin{thm}\label{thm:sparspromstabilitythm0}
	Let $n\geq 2$ and $\sigma \leq m_{\min}$. Let the measurement $\vect Y$ in (\ref{equ:modelsetting1}) be generated by a $n$-sparse measure $\mu=\sum_{j=1}^{n}a_j \delta_{\vect y_j}, \vect y_j \in B_{\frac{(2n-1)\pi}{12\Omega}, \infty}(\vect 0)$ in the two-dimensional space. Assume that
	\begin{equation}\label{equ:sparspromstabilitythm0equ1}
	D_{\min}:=\min_{p\neq j}\Big|\Big|\mathbf y_p-\mathbf y_j\Big|\Big|_1\geq \frac{15.3\pi(n-\frac{1}{2})}{\Omega }\Big(\frac{\sigma}{m_{\min}}\Big)^{\frac{1}{2n-1}}. \end{equation}
	Let $\mathcal O$ in the minimization problem (\ref{equ:l0minimization}) be  $B_{\frac{(n-1)\pi}{6\Omega}, \infty}(\vect 0)$, then the solution to (\ref{equ:l0minimization}) contains exactly $n$ point sources. For any solution $\hat \mu=\sum_{j=1}^{n}\hat a_{j}\delta_{\mathbf{\hat y}_j}$, it is in a $\frac{D_{\min}}{2}$-neighborhood of $\mu$. Moreover, after reordering the $\mathbf{\hat y}_j$'s, we have
	\begin{equation}
	\Big|\Big|\mathbf {\hat y}_j-\mathbf y_j\Big|\Big|_1\leq \frac{C(n)}{\Omega}SRF^{2n-2}\frac{\sigma}{m_{\min}}, \quad 1\leq j\leq n,
	\end{equation}
	where $SRF:=\frac{\pi}{D_{\min}\Omega} $ and $$C(n)=\frac{(1+\sqrt{3})^{2n-1}2^{5n-1}(2n-1)^{2n-1}\pi}{3^{2n-0.5}}.$$
\end{thm}

Theorem \ref{thm:sparspromstabilitythm0} reveals that sparsity promoting over admissible solutions could resolve the source locations to the resolution limit level. It provides an insight that theoretically sparsity-promoting algorithms would have excellent performance on the two-dimensional super-resolution problems. Especially, under the separation condition (\ref{equ:sparspromstabilitythm0equ1}), any tractable sparsity-promoting algorithms (such as total variation minimization algorithms \cite{candes2014towards}) rendering the sparsest solution
could stably reconstruct  all the
source locations. 

\section{Proofs of the main results}
The idea for proving the main results of the paper is to use some new techniques to reduce the two-dimensional problem to a one-dimensional case. The reduction techniques are mainly based on the three crucial observations in the following subsection. The estimation methods for the one-dimensional super-resolution problem are based  on a  nonlinear approximation theory in Vandermonde space, which we present in Section \ref{section:approxtheoryinvanderspace}. 

\subsection{Three crucial observations}\label{section:threeobservations}
We here introduce three crucial observations that reduce the two-dimensional super-resolution problem to its one-dimensional analog, by which we are able to derive the resolution limit theory of this paper. Our observations also pave the way for extending the resolution estimates to higher dimensions. Moreover, they inspire a new direction for the DOA algorithms; see  Sections \ref{section:twodnumberdetectalgo} and \ref{section:twodlocationrecoveryalgo}. 

\medskip
\noindent \textbf{Translation invariant:}\\
By the translation invariant we mean that if a measure $\hat \mu = \sum_{j=1}^q \hat a_j \delta_{\vect {\hat y}_j}$ is a $\sigma$-admissible measure for the measurement $\vect Y$, then $\hat \mu = \sum_{j=1}^q \hat a_j \delta_{\vect {\hat y}_j+\vect v}$ is a $\sigma$-admissible measure for measurement $e^{i\vect v^\top \vect \omega} \vect Y(\vect \omega)$ for any vector $\vect v\in \mathbb R^2$. More precisely, we have
\begin{equation}\label{equ:measureconstraint1}
\babs{\sum_{j=1}^q \hat a_j e^{i (\vect{\hat y}_j+\vect v )^\top \vect \omega} - e^{i\vect v^\top \vect \omega} \vect Y(\vect \omega)}= \babs{\sum_{j=1}^q \hat a_j e^{i \vect{\hat y}_j^\top \vect \omega} -  \vect Y(\vect \omega)}< \sigma, \quad \vect \omega \in [0, \Omega]^2.
\end{equation}
In addition, if for certain $\delta\geq 0$, 
\begin{equation}\label{equ:measureconstraint2}
\babs{\sum_{j=1}^q \hat a_j e^{i \vect{\hat y}_j^\top \vect \omega} - \sum_{j=1}^n a_j e^{i \vect{y}_j^\top \vect \omega}}< \delta, \quad \vect \omega \in [0, \Omega]^2,
\end{equation}
then for any vector $\vect v\in \mathbb R^2$, 
\[
\babs{\sum_{j=1}^q \hat a_j e^{i (\vect{\hat y}_j+\vect v )^\top \vect \omega} - \sum_{j=1}^n a_j e^{i (\vect{y}_j+\vect v)^\top \vect \omega}}< \delta, \quad \vect \omega \in [0, \Omega]^2.
\]

\medskip
\noindent\textbf{Combination of coordinates:}\\
The second observation is that if we suppose that (\ref{equ:measureconstraint2}) is satisfied, we have a similar estimate for the summation of combinations of $e^{i \tau \vect{\hat y}_{j,1}}, e^{i \tau \vect{\hat y}_{j,2}}$ and $e^{i \tau \vect{ y}_{j,1}}, e^{i \tau \vect{ y}_{j,2}}$ for certain $\tau$. Specifically, we have the following lemma.  

\begin{lem}\label{lem:coordinatecombination}
For any integer $t\geq 0$ and $\tau \leq \frac{\Omega}{t}$, the measurement constraint (\ref{equ:measureconstraint2}) implies
\[
\babs{\sum_{j=1}^q \hat a_j (e^{ir_1}e^{i \tau \vect{\hat y}_{j,1}}+e^{ir_2}e^{i \tau \vect{\hat y}_{j,2}})^t - \sum_{j=1}^n a_j (e^{ir_1}e^{i \tau \vect{y}_{j,1}}+e^{ir_2}e^{i \tau \vect{y}_{j,2}})^t} < 2^{t}\delta, \quad r_1, r_2\in \mathbb R. 
\] 
\end{lem}
\begin{proof} 
Let $\hat d_j = e^{ir_1}e^{i\tau \vect {\hat y}_{j,1}}+e^{ir_2}e^{i\tau \vect {\hat y}_{j,2}}$ and $d_j = e^{ir_1}e^{i\tau \vect {y}_{j,1}}+e^{ir_2}e^{i\tau \vect {y}_{j,2}}$. We have 
\begin{align*}
	&\babs{\sum_{j=1}^q \hat a_j \hat d_j^{t} -  \sum_{j=1}^n  a_j d_j^{t}}= \babs{\sum_{j=1}^q \hat a_j (e^{ir_1}e^{i \tau \vect {\hat y}_{j,1}}+e^{ir_2}e^{i\tau \vect {\hat y}_{j,2}})^{t}- \sum_{j=1}^n a_j (e^{ir_1}e^{i \tau \vect {y}_{j,1}}+e^{ir_2}e^{i\tau \vect {y}_{j,2}})^{t}}\\
	=&\babs{\sum_{t_1+t_2=t, 0\leq t_1,t_2\leq t} {t\choose t_1} \Big(\sum_{j=1}^q \hat a_j  e^{ir_1t_1}e^{ir_2t_2}e^{i \tau \vect {\hat y}_{j,1} t_1}e^{i\tau \vect {\hat y}_{j,2} t_2}- \sum_{j=1}^n a_je^{ir_1t_1}e^{ir_2t_2}  e^{i \tau \vect {y}_{j,1} t_1} e^{i\tau \vect {y}_{j,2}t_2}\Big)}\\
	\leq &\sum_{t_1+t_2=t, 0\leq t_1,t_2\leq t} {t\choose t_1} \babs{\sum_{j=1}^q \hat a_j  e^{i \tau \vect {\hat y}_{j,1} t_1}e^{i\tau \vect {\hat y}_{j,2} t_2}- \sum_{j=1}^n a_j  e^{i \tau \vect {y}_{j,1} t_1} e^{i\tau \vect {y}_{j,2}t_2}}\\
	=  &\sum_{t_1+t_2=t, 0\leq t_1,t_2\leq t} {t\choose t_1} \babs{\sum_{j=1}^q \hat a_j  e^{i (t_1\tau , t_2\tau)\vect {\hat y}_{j}}- \sum_{j=1}^n a_j  e^{i (t_1\tau , t_2\tau)\vect {y}_{j}}}\\
	<& \sum_{t_1+t_2=t, 0\leq t_1,t_2\leq t} {t\choose t_1} \delta \quad \Big(\text{by $\tau \leq  \frac{\Omega}{t}$ and (\ref{equ:measureconstraint2})}\Big)\\
	= & 2^{t}\delta.    
\end{align*}
\end{proof}

This is the key observation of the paper. It reduces the two-dimensional super-resolution problem to nearly a one-dimensional super-resolution one. Since it is about the difference between summation of combinations of $e^{i \tau \vect{\hat y}_{j,1}}, e^{i \tau \vect{\hat y}_{j,2}}$ and $e^{i \tau \vect{ y}_{j,1}}, e^{i \tau \vect{ y}_{j,2}}$, we refer to this observation as combination of coordinates and call the elements $e^{i \tau \vect{ y}_{j,1}}+e^{i \tau \vect{ y}_{j,2}}$  coordinate-combined elements. This coordinate-combination technique will be used in deriving new algorithms for the DOA problem in Sections \ref{section:twodnumberdetectalgo} and \ref{section:twodlocationrecoveryalgo}. 

Compared to the projection techniques in \cite{liu2021mathematicalhighd, chen2021algorithmic} which utilize the measurement constraint only in several one-dimensional spaces to derive stability results, our formulation utilizes more measurement constraints and consequently yields better estimates.

\medskip
\noindent \textbf{Preservation of the separation distance for the coordinate-combined elements:}\\
The last observation is that, for $\vect \theta_j$'s in $[0, \frac{2\pi}{3}]^2$, the coordinate-combined elements $e^{i \vect \theta_{j,1}}+ e^{i \vect \theta_{j,2}}$ still preserve the separation distance between the $\vect \theta_j$'s. This is revealed by  Lemma \ref{lem:projectiondislem1}. Note that the projection trick in \cite{liu2021mathematicalhighd, chen2021algorithmic} and many conventional two-dimensional DOA algorithms do not preserve the separation distance between the original source. This causes many issues in the reconstruction and resolution estimation. Lemma \ref{lem:projectiondislem1} is the main result of this paper by which  we could overcome the above issues and hence find a new way to solve two-dimensional DOA problems.

\begin{lem}\label{lem:projectiondislem1}
	For two different vectors $\vect \theta_j \in \mathbb [0, \frac{2\pi}{3}]^2, j=1,2$ with $\frac{\pi}{3}\leq  \vect \theta_{j,2} - \vect \theta_{j,1}\leq \frac{2}{3}\pi, j=1,2$, if $||\vect \theta_1 - \vect \theta_2||_1\geq \Delta$, then 
	\[
	\babs{e^{i \vect \theta_{1,1}}+ e^{i \vect \theta_{1,2}}- (e^{i \vect \theta_{2,1}} + e^{i \vect \theta_{2,2}})}\geq \frac{3}{2\pi}\Delta.
	\]
\end{lem}
\begin{proof}
	Note that $0\leq \vect \theta_{j,1}<\vect \theta_{j,2}\leq \frac{2\pi}{3}, j=1,2$. We prove the lemma by considering the following two cases.\\
	\textbf{Case 1:} $0\leq \vect \theta_{1,1}  \leq \vect \theta_{2,1} <  \vect \theta_{2,2}\leq  \vect \theta_{1,2}\leq \frac{2\pi}{3}$.\\
	In this case, 
	\begin{align*}
	\babs{e^{i \vect \theta_{1,1}}+ e^{i \vect \theta_{1,2}}- (e^{i \vect \theta_{2,1}} + e^{i \vect \theta_{2,2}})} \geq & \babs{e^{i \vect \theta_{2,1}}+ e^{i\vect \theta_{2,2}}}- \babs{e^{i \vect \theta_{1,1}} + e^{i\vect \theta_{1,2}}}\\
		= &2\Big(\cos(\frac{\phi_2}{2}) - \cos(\frac{\phi_1}{2}) \Big),
	\end{align*}
	where $\phi_j = \vect \theta_{j,2} - \vect \theta_{j,1}, j=1,2$. By the assumption made in the lemma, we have $\Delta \leq \phi_1 -\phi_2 \leq \frac{\pi}{3}$. Note also that $\frac{\pi}{6}\leq \frac{\phi_1+\phi_2}{4}\leq \frac{\pi}{3}$. Thus 
	\[
	2\Big(\cos(\frac{\phi_2}{2}) - \cos(\frac{\phi_1}{2}) \Big) =4\sin(\frac{\phi_1+\phi_2}{4})\sin(\frac{\phi_1-\phi_2}{4})\geq 4 \sin(\frac{\pi}{6})\sin(\frac{\Delta}{4})\geq  \frac{3\Delta}{2\pi}, 
	\]
	where the last inequality uses $\sin(\frac{\Delta}{4})\geq \frac{3}{\pi} \frac{\Delta}{4}$ for $0<\frac{\Delta}{4}\leq \frac{\pi}{12}$.\\
	\textbf{Case 2:} $0\leq \vect \theta_{1,1}  \leq  \vect \theta_{2,1} \leq   \vect \theta_{1,2}\leq  \vect \theta_{2,2}\leq \frac{2\pi}{3}$.\\
	The idea is to calculate the angle between $e^{i \vect \theta_{1,1}}+ e^{i \vect \theta_{1,2}}$ and $e^{i \vect \theta_{2,1}} + e^{i \vect \theta_{2,2}}$. By simple analysis of the angle relations between $e^{i \vect \theta_{1,1}}, e^{i \vect \theta_{1,2}}, e^{i \vect \theta_{2,1}}, e^{i \vect \theta_{2,2}}$, we obtain that the angle between $e^{i \vect \theta_{1,1}}+ e^{i \vect \theta_{1,2}}$ and $e^{i \vect \theta_{2,1}} + e^{i \vect \theta_{2,2}}$ is $\frac{\vect \theta_{2,1}-\vect \theta_{1,1}+ \vect \theta_{2,2}-\vect \theta_{1,2}}{2}$, which is larger than $\frac{\Delta}{2}$. Thus 
	\[
	\babs{e^{i \vect \theta_{1,1}}+ e^{i \vect \theta_{1,2}}- (e^{i \vect \theta_{2,1}} + e^{i \vect \theta_{2,2}})} \geq \max\Big(\babs{e^{i \vect \theta_{1,1}}+ e^{i \vect \theta_{1,2}}}, \babs{e^{i \vect \theta_{2,1}} + e^{i \vect \theta_{2,2}}}\Big)\sin(\frac{\Delta}{2}). 
	\]
	Since $\frac{\pi}{3}\leq  \vect \theta_{j,2} - \vect \theta_{j,1}\leq \frac{2}{3}\pi, j=1,2$,  we have 
	\[
	\max\Big(\babs{e^{i \vect \theta_{1,1}}+ e^{i \vect \theta_{1,2}}}, \babs{e^{i \vect \theta_{2,1}} + e^{i \vect \theta_{2,2}}}\Big) \geq 1.
	\]
	Therefore, 
	\[
	\babs{e^{i \vect \theta_{1,1}}+ e^{i \vect \theta_{1,2}}- (e^{i \vect \theta_{2,1}} + e^{i \vect \theta_{2,2}})}\geq \sin(\frac{\Delta}{2}) \geq \frac{3\Delta}{2\pi},
	\]
	where the last inequality uses $\sin(\frac{\Delta}{2})\geq \frac{3}{\pi}\frac{\Delta}{2}$ for $0<\frac{\Delta}{2}\leq \frac{\pi}{6}$.
\end{proof}

\subsection{Proof of Theorem \ref{thm:highdupperboundnumberlimit0}}
\begin{proof} The proof of this theorem is by contradiction. 
Suppose that there exists a measure $\hat \mu = \sum_{j=1}^q \hat a_j \delta_{\vect {\hat y}_j}$ with $q<n$ which is a $\sigma$-admissible measure of $\vect Y$. Then, by the measurement constraint (\ref{equ:modelsetting1}) and $|\vect W(\vect \omega)|<\sigma$,  we have 
\begin{equation}\label{equ:twodnumberproofequ1}
\babs{\sum_{j=1}^q \hat a_j e^{i \vect{\hat y}_j^\top \vect \omega} - \sum_{j=1}^n a_j e^{i \vect{y}_j^\top \vect \omega}}< 2\sigma, \quad \vect \omega \in [0, \Omega]^2.
\end{equation}
Since $\vect y_j\in [-\lambda, \lambda]^2$ with $\lambda= \frac{(n-1)\pi}{6\Omega}$, by letting $\vect v=(0,6\lambda)^\top$, we obtain
\begin{equation}
\vect y_j+\vect v\in [-\lambda,\lambda]\times[5\lambda, 7\lambda].
\end{equation} 
On the other hand, by (\ref{equ:twodnumberproofequ1}) we also get
\[
\babs{\sum_{j=1}^q \hat a_j e^{i (\vect{\hat y}_j+\vect v )^\top \vect \omega} - \sum_{j=1}^n a_j e^{i (\vect{y}_j+\vect v)^\top \vect \omega}}< 2\sigma, \quad \vect \omega \in [0, \Omega]^2.
\]
Thus with a slight abuse of notation, we still denote those $\vect{\hat y}_j+\vect v$ and $\vect{y}_j+\vect v$ by $\vect {\hat y}_j$, $\vect y_j$ respectively and consider them in the rest of the proof. Note that we have 
\[
\vect y_j \in [-\lambda,\lambda]\times[5\lambda, 7\lambda], \quad j =1, \cdots, n.
\]
Let $\tau = \frac{\Omega}{2(n-1)}$, together with $\lambda = \frac{(n-1)\pi}{6\Omega}$, we have $\tau \vect y_j \in [-\frac{\pi}{12},\frac{\pi}{12}]\times[\frac{5\pi}{12}, \frac{7\pi}{12}]$. 
This yields
\begin{equation}\label{equ:twodnumberproofequ2}
-\frac{\pi}{12}\leq \tau \vect  y_{j,1}\leq \frac{\pi}{12},\quad  \frac{5\pi}{12}\leq \tau \vect  y_{j,2}\leq \frac{7\pi}{12}, \quad  \frac{\pi}{3}\leq \tau \vect  y_{j,2}- \tau \vect  y_{j,1} \leq \frac{2\pi}{3}. 
\end{equation}
On the other hand, let $\hat d_j = e^{i\tau \vect {\hat y}_{j,1}}+e^{i\tau \vect {\hat y}_{j,2}}$ and $d_j = e^{i\tau \vect {y}_{j,1}}+e^{i\tau \vect {y}_{j,2}}$. By Lemma \ref{lem:coordinatecombination} and (\ref{equ:twodnumberproofequ1}) we have that 
\begin{equation}\label{equ:twodnumberproofequ4}
\babs{\sum_{j=1}^q \hat a_j \hat d_j^{t} -  \sum_{j=1}^n  a_j d_j^{t}}< 2^{t+1} \sigma,  \quad t =0, 1, \cdots, 2n-2. 
\end{equation}
Let 
\[
\vect b = \Big(\sum_{j=1}^q \hat a_j \hat d_j^{0} -  \sum_{j=1}^n  a_j d_j^{0}, \quad \sum_{j=1}^q \hat a_j \hat d_j^{1} -  \sum_{j=1}^n  a_j d_j^{1}, \quad \cdots, \quad \sum_{j=1}^q \hat a_j \hat d_j^{2n-2} -  \sum_{j=1}^n  a_j d_j^{2n-2}\Big)^\top.
\]
Since (\ref{equ:twodnumberproofequ2}) holds,  Lemma \ref{lem:projectiondislem1} yields
\[
d_{\min}:=\min_{p\neq q}\babs{d_p - d_q}\geq \frac{3}{2\pi}
\min_{p\neq q}\tau \Big|\Big|\vect y_p - \vect y_q\Big|\Big|_1 > 12.4 \Big(\frac{\sigma}{m_{\min}}\Big)^{\frac{1}{2n-2}}> 2\sqrt{6(1+\sqrt{3})}\Big(\frac{4}{\sqrt{3}} \frac{\sigma}{m_{\min}}\Big)^{\frac{1}{2n-2}},
\]
where the second last inequality is due to the separation condition (\ref{equ:highdupperboundnumberlimit1}). On the other hand, we have $|\hat d_p|\leq 2, p=1, \cdots, q$ and $|d_j|\leq \sqrt{3}, j=1, \cdots, n$ since (\ref{equ:twodnumberproofequ2}) holds. Thus we can apply Theorem \ref{thm:spaceapproxlowerbound1} and get 
\[
||\vect b||_2 \geq \frac{m_{\min}(d_{\min})^{2n-2}}{(2(1+2)(1+\sqrt{3}))^{(n-1)}} > \frac{4^n \sigma}{\sqrt{3}}.
\]  
However, (\ref{equ:twodnumberproofequ4}) implies that $||\vect b||_2 < \frac{4^n \sigma}{\sqrt{3}}$, which is a contradiction. This proves the theorem. 
\end{proof}

\subsection{Proof of Theorem \ref{thm:highdupperboundsupportlimit0}}
\begin{proof}
Note that $\vect y_j, \vect {\hat y}_j$'s are in $[-\lambda, \lambda]^2$ with $\lambda = \frac{(2n-1)\pi}{12\Omega}$ and $\hat \mu = \sum_{j=1}^n \hat a_j \delta_{\vect {\hat y}_j}$ is a $\sigma$-admissible measure of $\vect Y$. Let $\tau = \frac{\Omega}{2n-1}$. Similarly to the proof in the above section, we can construct $\vect x_j = \vect y_j+\vect v, \vect {\hat x}_j = \vect {\hat y}_j+\vect v$ so that $\tau \vect {\hat x}_j, \tau \vect x_j \in [-\frac{\pi}{12},\frac{\pi}{12}]\times[\frac{5\pi}{12}, \frac{7\pi}{12}]$ and 
\begin{equation}\label{equ:twodsupportproofequ1}
\babs{\sum_{j=1}^n \hat a_j e^{i \vect{\hat x}_j^\top \vect \omega} - \sum_{j=1}^n a_j e^{i \vect{x}_j^\top \vect \omega}}< 2\sigma, \quad \vect \omega \in [0, \Omega]^2.
\end{equation}
 Thus we have 
\begin{equation}\label{equ:twodsupportproofequ2}
-\frac{\pi}{12}\leq \tau \vect  x_{j,1}\leq \frac{\pi}{12},\quad  \frac{5\pi}{12}\leq \tau \vect  x_{j,2}\leq \frac{7\pi}{12}, \quad  \frac{\pi}{3}\leq \tau \vect  x_{j,2}- \tau \vect  x_{j,1} \leq \frac{2\pi}{3},
\end{equation}
\begin{equation}\label{equ:twodsupportproofequ3}
-\frac{\pi}{12}\leq \tau \vect {\hat x}_{j,1}\leq \frac{\pi}{12},\quad  \frac{5\pi}{12}\leq \tau \vect  {\hat x}_{j,2}\leq \frac{7\pi}{12}, \quad  \frac{\pi}{3}\leq \tau \vect  {\hat x}_{j,2}- \tau \vect  {\hat x}_{j,1} \leq \frac{2\pi}{3}.
\end{equation}
Moreover, it follows that
\begin{equation}\label{equ:twodsupportproofequ4}
-\frac{\pi}{12}\leq \tau \vect  x_{j,1}\leq \frac{\pi}{12},\quad  \frac{-7\pi}{12}\leq \tau \vect  x_{j,2}-\pi\leq \frac{-5\pi}{12}, \quad  \frac{\pi}{3}\leq  \tau \vect  x_{j,1}- (\tau \vect x_{j,2}-\pi) \leq \frac{2\pi}{3},
\end{equation}
\begin{equation}\label{equ:twodsupportproofequ5}
-\frac{\pi}{12}\leq \tau \vect  {\hat x}_{j,1}\leq \frac{\pi}{12},\quad  \frac{-7\pi}{12}\leq \tau \vect  {\hat x}_{j,2}-\pi\leq \frac{-5\pi}{12}, \quad  \frac{\pi}{3}\leq  \tau \vect  {\hat x}_{j,1}- (\tau \vect {\hat x}_{j,2}-\pi) \leq \frac{2\pi}{3}. 
\end{equation}
Let $\hat d_j = e^{i\tau \vect {\hat x}_{j,1}}+e^{i\tau \vect {\hat x}_{j,2}}, d_j = e^{i\tau \vect {x}_{j,1}}+e^{i\tau \vect {x}_{j,2}}$ and $\hat g_j = e^{i\tau \vect {\hat x}_{j,1}}+e^{i(\tau \vect {\hat x}_{j,2}-\pi)}, g_j = e^{i\tau \vect {x}_{j,1}}+e^{i(\tau \vect {x}_{j,2}-\pi)}$. By (\ref{equ:twodsupportproofequ1}) and Lemma \ref{lem:coordinatecombination}, we arrive at
\begin{align}
&\babs{\sum_{j=1}^n \hat a_j \hat d_j^{t} -  \sum_{j=1}^n  a_j d_j^{t}}< 2^{t+1} \sigma, \quad t=0,1, \cdots, 2n-1, \label{equ:twodsupportproofequ6} \\ 
&\babs{\sum_{j=1}^n \hat a_j \hat g_j^{t} -  \sum_{j=1}^n  a_j g_j^{t}}< 2^{t+1} \sigma,  \quad t =0, 1, \cdots, 2n-1. \label{equ:twodsupportproofequ7} 
\end{align}
 Let 
\[
\vect d = \Big(\sum_{j=1}^n \hat a_j \hat d_j^{0} -  \sum_{j=1}^n  a_j d_j^{0}, \quad \sum_{j=1}^n \hat a_j \hat d_j^{1} -  \sum_{j=1}^n  a_j d_j^{1}, \quad \cdots,\quad \sum_{j=1}^n \hat a_j \hat d_j^{2n-1} -  \sum_{j=1}^n  a_j d_j^{2n-1}\Big)^\top,
\]
and 
\[
\vect g = \Big(\sum_{j=1}^n \hat a_j \hat g_j^{0} -  \sum_{j=1}^n  a_j g_j^{0}, \quad \sum_{j=1}^n \hat a_j \hat g_j^{1} -  \sum_{j=1}^n  a_j g_j^{1}, \quad \cdots,\quad \sum_{j=1}^n \hat a_j \hat g_j^{2n-1} -  \sum_{j=1}^n  a_j g_j^{2n-1}\Big)^\top.
\]
Equations (\ref{equ:twodsupportproofequ6}) and (\ref{equ:twodsupportproofequ7}) imply respectively 
\[
||\vect d||_2 < \frac{2^{2n+1} \sigma}{\sqrt{3}}, \quad ||\vect g||_2 < \frac{2^{2n+1} \sigma}{\sqrt{3}}. 
\]
Note also that by (\ref{equ:twodsupportproofequ2}), (\ref{equ:twodsupportproofequ3}), (\ref{equ:twodsupportproofequ4}), and (\ref{equ:twodsupportproofequ5}), we get
\[|\hat d_j|, |d_j|, |\hat g_j|, |g_j|\leq \sqrt{3},\quad  j=1, \cdots, n.
\]
Define $d_{\min}:=\min_{p\neq q}\babs{d_p - d_q}$ and $g_{\min}:=\min_{p\neq q}\babs{g_p - g_q}$.  Applying Theorem \ref{thm:spaceapproxlowerbound1}, we thus have that
\begin{equation}\label{equ:twodsupportproofequ8}
\Big|\Big|\eta_{n,n}(d_1,\cdots, d_n, \hat d_1, \cdots, \hat d_n)\Big|\Big|_{\infty}<\frac{(1+\sqrt{3})^{2n-1}}{ d_{\min}^{n-1}} \frac{2^{2n+1} }{\sqrt{3}}\frac{\sigma}{m_{\min}}, 	    
\end{equation}
and 
\begin{equation}\label{equ:twodsupportproofequ9}
\Big|\Big|\eta_{n,n}(g_1,\cdots, g_n, \hat g_1, \cdots, \hat g_n)\Big|\Big|_{\infty}<\frac{(1+\sqrt{3})^{2n-1}}{ g_{\min}^{n-1}} \frac{2^{2n+1} }{\sqrt{3}}\frac{\sigma}{m_{\min}}. 	    
\end{equation}
We now demonstrate that we can reorder $\hat d_j, \hat g_j$ to have $|\hat d_j -d_j|< \frac{d_{\min}}{2}$ and $|\hat g_j -g_j|< \frac{g_{\min}}{2}, j=1, \cdots, n$. First, since (\ref{equ:twodsupportproofequ2}) and (\ref{equ:twodsupportproofequ4}) hold, by Lemma \ref{lem:projectiondislem1} we have 
\begin{align}\label{equ:twodsupportproofequ10}
d_{\min}\geq \frac{3}{2\pi} \min_{p\neq q}\tau \Big|\Big|\vect y_p - \vect y_q\Big|\Big|_1 \geq 11.475 \Big(\frac{\sigma}{m_{\min}}\Big)^{\frac{1}{2n-1}}> 2^{3/2}(1+\sqrt{3})\Big(\frac{2^{5/2}}{\sqrt{3}} \frac{\sigma}{m_{\min}}\Big)^{\frac{1}{2n-1}},
\end{align}
and 
\[
g_{\min}\geq \frac{3}{2\pi} \min_{p\neq q}\tau \Big|\Big|\vect y_p - \vect y_q\Big|\Big|_1 \geq 11.475 \Big(\frac{\sigma}{m_{\min}}\Big)^{\frac{1}{2n-1}}> 2^{3/2}(1+\sqrt{3})\Big(\frac{2^{5/2}}{\sqrt{3}} \frac{\sigma}{m_{\min}}\Big)^{\frac{1}{2n-1}},
\]
where we also use separation condition (\ref{equ:highdsupportlimithm0equ0}) in the above derivation. Let 
\[
\epsilon_d = \frac{(1+\sqrt{3})^{2n-1}}{ d_{\min}^{n-1}} \frac{2^{2n+1}}{\sqrt{3}}\frac{\sigma}{m_{\min}},\quad  \epsilon_g = \frac{(1+\sqrt{3})^{2n-1}}{ g_{\min}^{n-1}}\frac{2^{2n+1} }{\sqrt{3}}\frac{\sigma}{m_{\min}}.
\]
By (\ref{equ:twodsupportproofequ10}), we have  
\[
d_{\min}^{2n-1}\geq \frac{(1+\sqrt{3})^{2n-1}2^{3n+1}}{\sqrt{3}} \frac{\sigma}{m_{\min}}, \ \text{or equivalently},  d_{\min}^{n}\geq 2^n \epsilon_d.
\]
A similar result holds for $g_{\min}$ and $\epsilon_g$. Thus the conditions of Lemma \ref{lem:multiproductstability1} are satisfied. By Lemma \ref{lem:multiproductstability1}, we have that after reordering $\hat d_j , \hat g_j$, 
\[
\Big|\hat d_j- d_j\Big|< \frac{d_{\min}}{2},\quad  \Big|\hat g_j- g_j\Big|< \frac{g_{\min}}{2},
\]
and
\begin{align*}
&\babs{\hat d_j -d_j}\leq \Big(\frac{2}{d_{\min}}\Big)^{n-1}\epsilon_d = \Big(\frac{1}{d_{\min}}\Big)^{2n-2}\frac{(1+\sqrt{3})^{2n-1}2^{3n}}{\sqrt{3}}\frac{\sigma}{m_{\min}},\\
&\babs{\hat g_j -g_j}\leq \Big(\frac{2}{g_{\min}}\Big)^{n-1}\epsilon_g= \Big(\frac{1}{g_{\min}}\Big)^{2n-2} \frac{(1+\sqrt{3})^{2n-1}2^{3n}}{\sqrt{3}}\frac{\sigma}{m_{\min}}. 
\end{align*}
Observing 
\[
e^{i\tau \vect {\hat x}_{j,1}}- e^{i\tau \vect {x}_{j, 1}} =  \frac{1}{2} \Big(\hat d_j -d_j+ \hat g_j -g_j\Big),\qquad  e^{i\tau\vect {\hat x}_{j,2}}- e^{i \tau \vect {x}_{j, 2}} =  \frac{1}{2} \Big(\hat d_j -d_j- (\hat g_j -g_j)\Big),
\]
we conclude that 
\[
\Big|e^{i\tau\vect {\hat x}_{j,1}}- e^{i\tau\vect {x}_{j, 1}}\Big|+\Big|e^{i\tau\vect {\hat x}_{j,2}}- e^{i\tau\vect {x}_{j, 2}}\Big|\leq \Big(\Big(\frac{1}{d_{\min}}\Big)^{2n-2} + \Big(\frac{1}{g_{\min}}\Big)^{2n-2}\Big) \frac{(1+\sqrt{3})^{2n-1}2^{3n}}{\sqrt{3}}\frac{\sigma}{m_{\min}}.
\]
On the other hand, by (\ref{equ:twodsupportproofequ2}) and (\ref{equ:twodsupportproofequ3}), $$|\vect {\hat x}_{j,1}- \vect {x}_{j, 1}|\leq \frac{\pi}{6} \quad \mbox{and} \quad  |\vect {\hat x}_{j,2}- \vect {x}_{j, 2}|\leq \frac{\pi}{6}.$$ We further have 
\begin{align*}
&\tau\Big|\vect {\hat x}_{j,1}- \vect {x}_{j, 1}\Big|+ \tau\Big|\vect {\hat x}_{j,2}- \vect x_{j, 2}\Big| \leq  \frac{\pi}{3}\Big(\Big|e^{i\vect {\hat x}_{j,1}}- e^{i\vect {x}_{j, 1}}\Big|+\Big|e^{i\vect {\hat x}_{j,2}}- e^{i\vect {x}_{j, 2}}\Big|\Big)\\
\leq& \Big(\Big(\frac{1}{d_{\min}}\Big)^{2n-2} + \Big(\frac{1}{g_{\min}}\Big)^{2n-2}\Big) \frac{(1+\sqrt{3})^{2n-1}2^{3n}\pi}{3\sqrt{3}}\frac{\sigma}{m_{\min}}.
\end{align*}
Recalling that $\tau=\frac{\Omega}{2n-1}$, we have \begin{align*}
&\Big|\vect {\hat x}_{j,1}- \vect {x}_{j, 1}\Big|+ \Big|\vect {\hat x}_{j,2}- \vect x_{j, 2}\Big| \leq \frac{2n-1}{\Omega} \Big(\Big(\frac{1}{d_{\min}}\Big)^{2n-2} + \Big(\frac{1}{g_{\min}}\Big)^{2n-2}\Big) \frac{(1+\sqrt{3})^{2n-1}2^{3n}\pi}{3\sqrt{3}}\frac{\sigma}{m_{\min}}.
\end{align*}
Note that by (\ref{equ:twodsupportproofequ10}), we obtain that $$D_{\min}\leq \frac{2\pi(2n-1)}{3\Omega}d_{\min} \quad \mbox{and} \quad  D_{\min}\leq \frac{2\pi(2n-1)}{3\Omega}g_{\min}.$$ Thus 
\begin{align*}
\bonenorm{\vect {\hat x}_j-\vect {x}_j} \leq&  \frac{(1+\sqrt{3})^{2n-1}2^{3n+1}\pi(2n-1)}{3\sqrt{3}\Omega} \Big(\frac{2(2n-1)}{3}\Big)^{2n-2} \Big(\frac{\pi}{\Omega D_{\min}}\Big)^{2n-2} \frac{\sigma}{m_{\min}}\\
=& \frac{(1+\sqrt{3})^{2n-1}2^{5n-1}(2n-1)^{2n-1}\pi}{3^{2n-0.5}\Omega}  \Big(\frac{\pi}{\Omega D_{\min}}\Big)^{2n-2} \frac{\sigma}{m_{\min}}.
\end{align*}
Since $||\vect {\hat y}_{j}- \vect{y}_{j}||_1 = ||\vect {\hat x}_{j}- \vect {x}_{j}||_1 $, we further get 
\begin{align*}
\bonenorm{\vect {\hat y}_j-\vect {y}_j} \leq \frac{(1+\sqrt{3})^{2n-1}2^{5n-1}(2n-1)^{2n-1}\pi}{3^{2n-0.5}\Omega}  \Big(\frac{\pi}{\Omega D_{\min}}\Big)^{2n-2} \frac{\sigma}{m_{\min}}.
\end{align*}
Since $D_{\min}\geq \frac{15.3\pi(n-0.5)}{\Omega} \Big( \frac{\sigma}{m_{\min}}\Big)^{\frac{1}{2n-1}}$, together with the above estimate, we can also show that
\[
\bonenorm{\vect {\hat y}_j-\vect {y}_j}< \frac{D_{\min}}{2}. 
\]
This completes the proof.
\end{proof}

\section{An algorithm for the model order detection in two-dimensional DOA estimation}\label{section:twodnumberdetectalgo}
In this section, based on the observations made in Section \ref{section:threeobservations}, we propose a new algorithm, named coordinate-combination-based sweeping singular-value-thresholding number detection algorithm, for the model order detection in two-dimensional DOA estimations. 

\begin{figure}[!h]
	\includegraphics[width=0.9\textwidth]{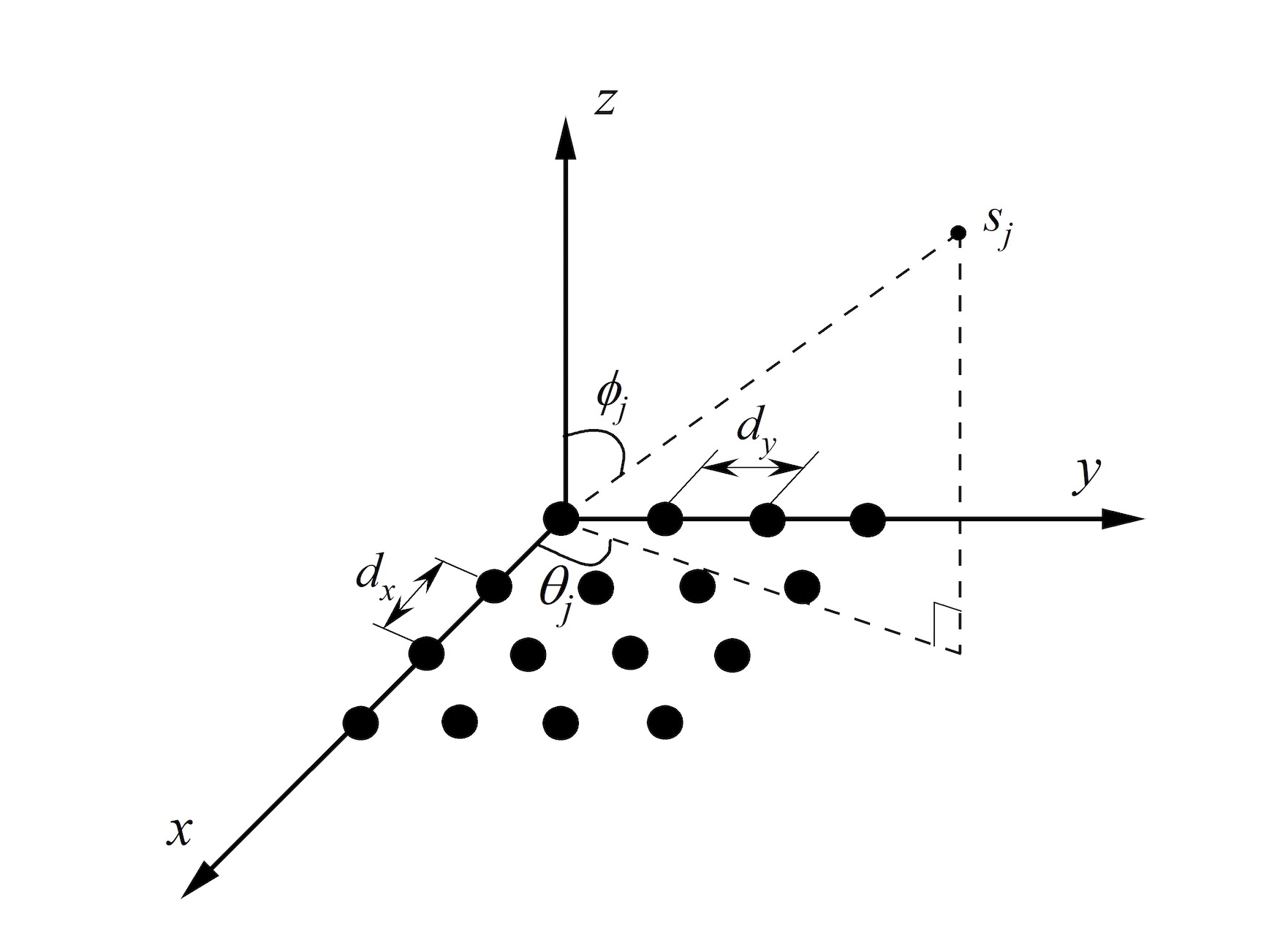}
	\centering
	\caption{The geometry of a uniform rectangular array.}
	\label{fig:geoofURA}
\end{figure}

\subsection{Problem formulation}
The existing two-dimensional DOA algorithms usually try to estimate the azimuth and elevation angles $(\theta_j, \phi_j)$'s that are shown in Figure \ref{fig:geoofURA}. More precisely, we consider $n$ narrowband signals/sources impinging on an $(\Omega+1)\times(\Omega+1)$ uniform rectangular array (URA) with $(\Omega+1)^2$ well calibrated and identically polarized antenna elements. The signal received by these antenna elements in a single snapshot can be expressed by 
\begin{equation}\label{equ:DOAmeasuresetting0}
\vect Y(\vect \omega) = \sum_{j=1}^n s_j p_j e^{j k d_x\vect \omega_1\vect y_{j,1}}e^{j k d_y\vect \omega_2\vect y_{j,2}}+\vect W(\vect \omega), \quad \vect{\omega}\in [0,1, \cdots, \Omega]^2,
\end{equation}
where $s_j$ is the $j$-th incident signal, $p_j$ is a complex constant denoting the signal/antenna polarization mismatch, $k$ represents the wavenumber of the carrier frequency, and $d_x$ and $d_y$ denote the distance between adjacent antenna element along the $x$-axis and $y$-axis, respectively. $\vect y_{j,1}=\sin \phi_j \cos \theta_j$ is the direction component of signal $s_j$ propagating along the $x$-axis and $\vect y_{j,2}=\sin \phi_j \sin \theta_j$ is the one propagating along the $y$-axis. The $\phi_j$ and $\theta_j$ denote respectively the elevation and azimuth angles of $s_j$. $\vect W(\vect \omega)$ is the additive noise, which is usually assumed to be white Gaussian noise.

For convenience, we consider the following simplified form of (\ref{equ:DOAmeasuresetting0}):
\begin{equation}\label{equ:DOAmeasuresetting1}
	\mathbf Y(\vect{\omega}) = \sum_{j=1}^{n}a_j e^{i \vect{y}_j^\top \vect{\omega}} + \mathbf W(\vect{\omega}),\qquad  \ \vect{\omega}\in [0,1, \cdots, \Omega]^2,
\end{equation}
where $\mathbf W$ is the noise with $||\mathbf W(\vect{\omega})||_\infty< \sigma$ and $\sigma$ being the noise level. We aim to recover stably the number of the signals and the $\vect y_j$'s, by which the elevation and azimuth angles are stably resolved. For a better exposition, we still consider a discrete measure $\mu = \sum_{j=1}^n a_j \delta_{\vect y_j}$ and denote the $a_j \delta_{\vect y_j}$'s as sources. The measurement (\ref{equ:DOAmeasuresetting1}) can be viewed as the noisy Fourier data of the measure $\mu$ at some discrete points. 

In this section and the next one, we shall propose new algorithms for detecting the model order and recovering the supports of $\mu$ from the measurement (\ref{equ:DOAmeasuresetting1}). Our number detection method is based on thresholding on a Hankel matrix assembled by data from modifications of (\ref{equ:DOAmeasuresetting1}) . The following subsection shall introduce the details of the Hankel matrix formulation. We refer to \cite{akaike1998information, akaike1974new, wax1985detection, schwarz1978estimating, rissanen1978modeling, wax1989detection, lawley1956tests, chen1991detection, he2010detecting, han2013improved, liu2021theorylse, liu2021mathematicalhighd} for other model detecting algorithms.

\subsection{Hankel matrix construction}\label{section:numberalgohankelconstruct}
The Hankel matrix is constructed by the following three steps.

\medskip
\noindent\textbf{Measurement modification by source translation}\\
Due to the translation invariance, suppose the sources are supported in $[-\lambda, \lambda]^2$, we consider them displacing with a vector $\vect v$ and get that $\vect x_j = \vect y_j +\vect v$. Using a simple measurement modification technique, we obtain the measurement for the new source $\tilde{\mu} = \sum_{j=1}^n a_j \delta_{\vect x_j}$. Specifically, we consider 
\begin{equation}\label{equ:numberdetectalgomeasure1}
\begin{aligned}
    \vect X(\vect \omega) = & e^{i \vect v^\top \vect \omega} \vect Y(\vect \omega) = \sum_{j=1}^{n}a_j e^{i (\vect{y}_j+\vect v)^\top \vect{\omega}} + e^{i \vect v^\top \vect \omega} \mathbf W(\vect{\omega})\\
    = & \sum_{j=1}^{n}a_j e^{i \vect x_j^\top \vect{\omega}} +  \vect  {\tilde{W}}(\vect{\omega}), \quad  \vect{\omega}\in [0,1, \cdots, \Omega]^2,
\end{aligned}
\end{equation}
with $|\mathbf {\tilde{W}}(\vect{\omega})|< \sigma$. 

\medskip
\noindent \textbf{Measurement modification by coordinate-combination}\\
The second procedure consists in modifying the measurement based on coordinate-combination. For  $s >0$, let $r = \frac{\Omega}{2s}$. From the measurement $\vect X$, we construct a list of new data given by
\[
\vect D(t) = \sum_{t_1+t_2 =t, 0\leq t_1, t_2\leq t} {t\choose t_1} \vect X(\vect \omega_{rt_1, rt_2}), \quad t=0, \cdots, 2s,
\]
where $\vect \omega_{rt_1, rt_2} = (rt_1, rt_2)^\top$. Note that 
\begin{align*}
\vect D(t) = & \sum_{j=1}^n a_j(e^{i \vect x_{j,1}r}+ e^{i \vect x_{j,2}r})^{t} + \sum_{t_1+t_2 =t, 0\leq t_1, t_2\leq t} {t\choose t_1}  \vect {\tilde{W}}(\vect \omega_{rt_1, rt_2})\\
= & \sum_{j=1}^na_j(e^{i \vect x_{j,1}r}+ e^{i \vect x_{j,2}r})^{t} + \vect {\hat W}(t), 
\end{align*}
where $\vect {\hat W}(t) = \sum_{t_1+t_2 =t, 0\leq t_1, t_2\leq t} {t\choose t_1}  \vect {\tilde{W}}(\vect \omega_{rt_1, rt_2})$. 

\medskip
\noindent \textbf{Hankel matrix construction and singular value decomposition}\\
Finally, from these $\vect D(t)$'s, we assemble the following Hankel matrix
\begin{equation}\label{equ:hankelmatrix1}
\mathbf H(s)=\begin{pmatrix}
\mathbf D(0) &\mathbf D(1)&\cdots& \mathbf D(s)\\
\mathbf D(1)&\mathbf D(2)&\cdots&\mathbf D(s+1)\\
\cdots&\cdots&\ddots&\cdots\\
\mathbf D(s)&\mathbf D(s+1)&\cdots&\mathbf D(2s)
\end{pmatrix}.
\end{equation}
We observe that $\mathbf H(s)$ has the decomposition
\begin{equation}\label{equ:hankeldecomp1}
\mathbf H(s)= BAB^T+\Delta,
\end{equation}
where $A=\text{diag}(a_1, \cdots, a_n)$ and $B=\big(\phi_{s}(e^{i \vect x_{j,1}r}+ e^{i \vect x_{j,2}r}), \cdots, \phi_{s}(e^{i \vect x_{j,1}r}+ e^{i \vect x_{j,2}r})\big)$ with $\phi_{s}(\omega)$ being defined as
\begin{equation}\label{equ:defineofphi0}
\phi_{s}(\omega) = (1, \omega, \cdots, \omega^s)^\top, 
\end{equation}
and 
\begin{equation}\label{equ:hankeldecompnoisematrix}
\Delta = \begin{pmatrix}
\mathbf {\hat W}(0)&\mathbf {\hat W}(1)&\cdots& \mathbf {\hat W}(s)\\
\mathbf {\hat W}(1)&\mathbf {\hat W}(2)&\cdots&\mathbf {\hat W}(s+1)\\
\vdots&\vdots&\ddots&\vdots\\
\mathbf {\hat W}(s)&\mathbf {\hat W}(s+1)&\cdots&\mathbf {\hat W}(2s)
\end{pmatrix}.
\end{equation}
We denote the singular value decomposition of $\mathbf H(s)$ as  
\[\mathbf H(s)=\hat U\hat \Sigma \hat U^*,\]
where $\hat\Sigma =\text{diag}(\hat \sigma_1,\cdots, \hat \sigma_n, \hat \sigma_{n+1},\cdots,\hat\sigma_{s+1})$ with the singular values $\hat \sigma_j$, $1\leq j \leq s+1$, ordered in a decreasing manner. The source number $n$ is then detected by a thresholding on these singular values. In the next subsection we will provide the theoretical guarantee of the threshold.

\subsection{Theoretical guarantee}
Note that when there is no noise, $\mathbf H(s)=BAB^{\top}$. We have the following estimate for the singular values of $BAB^{\top}$.

\begin{lem}\label{lem:numberdetectalgolem0}
	Let $n\geq 2$, $s\geq n$, $\vect y_j\in [-\frac{s\pi}{6\Omega}, \frac{s\pi}{6\Omega}]^2, 1\leq j\leq n,$ and $\vect v$ in (\ref{equ:numberdetectalgomeasure1}) be $(0, \frac{s\pi}{\Omega})^\top$. Let $$\sigma_1,\cdots,\sigma_n,0,\cdots,0$$ be the singular values of $BAB^T$ in (\ref{equ:hankeldecomp1}) ordered in a decreasing manner. Then the following estimate holds
	\begin{equation}\label{equ:numberdetectalgolem0equ1}
	\sigma_n\geq \frac{m_{\min}\big(3\theta_{\min}(\Omega,s)\big)^{2n-2}}{n(2(1+\sqrt{3})\pi)^{2n-2}},
	\end{equation}
	where $\theta_{\min}(\Omega, s)=\min_{p\neq j}\bonenorm{\vect y_p\frac{\Omega}{2s}-\vect y_j \frac{\Omega}{2s}}$.
\end{lem}
\begin{proof} 
Recall that $\sigma_n$ is the minimum nonzero singular value of $BAB^\top$. Let $\ker(B^\top)$ be the kernel space of $B^\top$ and $\ker^{\perp}(B^\top)$ be its orthogonal complement. Then we have 
\begin{align*}
&\sigma_n=\min_{||x||_2=1,x\in \ker^{\perp}(B^\top)}||BAB^\top x||_2\geq \sigma_{\min}(BA)\sigma_n(B^\top)\\
\geq& \sigma_{\min}(B)\sigma_{\min}(A)\sigma_{\min}(B).
\end{align*}
On the other hand, since by the condition of the lemma $\vect x_j = \vect y_j +\vect v \in [-\frac{s\pi}{6\Omega}, \frac{s\pi}{6\Omega}]\times [\frac{5s\pi}{6\Omega}, \frac{7s\pi}{6\Omega}]$, we have $ \frac{\Omega \vect x_j}{2s} \in [-\frac{\pi}{12}, \frac{\pi}{12}]\times [\frac{5\pi}{12}, \frac{7\pi}{12}]$. Thus, by Lemma \ref{lem:projectiondislem1}, for $r= \frac{\Omega}{2s}$, 
\[
\min_{p\neq q}\babs{e^{i \vect x_{p,1}r}+ e^{i \vect x_{p,2}r} -(e^{i \vect x_{q,1}r}+ e^{i \vect x_{q,2}r})} \geq \frac{3}{2\pi} \theta_{\min}(\Omega, s). 
\]
Note also that $|e^{i \vect x_{p,1}r}+ e^{i \vect x_{p,2}r}|\leq \sqrt{3}$. Thus applying Lemma \ref{lem:singularvaluevandermonde2} and Corollary \ref{lem:norminversevandermonde1}, we have
\begin{align*}
\sigma_{\min}(B)\geq \frac{1}{\sqrt{n}}\frac{\big(\frac{3}{2\pi}\theta_{\min}(\Omega,s)\big)^{n-1}}{(1+\sqrt{3})^{n-1}}.
\end{align*}
Then, it follows that
\begin{align*}\sigma_n\geq \sigma_{\min}(A)\Big(\frac{\big(\frac{3}{2\pi}\theta_{\min}(\Omega,s)\big)^{n-1}}{(1+\sqrt{3})^{n-1}}\Big)^2
\geq \frac{m_{\min}\big(3\theta_{\min}(\Omega,s)\big)^{2n-2}}{n(2(1+\sqrt{3})\pi)^{2n-2}}.\end{align*}
\end{proof}

\vspace{0.3cm}
We now present the main result on the threshold for the singular values of the matrix $\mathbf{H}(s)$.
\begin{thm}\label{thm:numberdetectalgothm1}
	Let $n\geq 2, s\geq n$ and $\mu=\sum_{j=1}^{n}a_j \delta_{\vect y_j}$ with $\vect y_j\in [-\frac{s\pi}{6\Omega}, \frac{s\pi}{6\Omega}]^2, 1\leq j\leq n$. Let $\vect v$ in (\ref{equ:numberdetectalgomeasure1}) be equal to $(0, \frac{s\pi}{\Omega})^\top$. Then for the singular values of $\vect H(s)$ in (\ref{equ:hankelmatrix1}), We have 
	\begin{equation}\label{equ:numberdetectalgothm1equ0}
	\hat \sigma_j< \frac{4^{s+1}\sigma}{3},\quad j=n+1,\cdots,s+1.
	\end{equation}
	Moreover, if the following separation condition is satisfied
	\begin{equation}\label{equ:numberdetectalgothm1equ1}
	\min_{p\neq j}\bonenorm{\vect y_p-\vect y_j}\geq \frac{4(1+\sqrt{3})\pi s }{3\Omega}\Big(\frac{2n4^{s+1}}{3}\frac{\sigma}{m_{\min}}\Big)^{\frac{1}{2n-2}},
	\end{equation}
	then
	\begin{equation}\label{MUSICthm1equ2}
	\hat\sigma_{n}\geq \frac{4^{s+1}\sigma}{3}.
	\end{equation}
\end{thm}
\begin{proof}
We first estimate $||\Delta||_2$ for $\Delta$ in (\ref{equ:hankeldecompnoisematrix}). By the definition of $\vect {\hat W}(t)$, we have $|\vect {\hat W}(t)< 2^{t} \sigma$. Thus $||\Delta||_2\leq ||\Delta||_F< \frac{4^{s+1}\sigma}{3}$. By Weyl's theorem, we have $|\hat \sigma_j-\sigma_j|\leq ||\Delta||_2, j=1,\cdots,n$. Together with $\sigma_j=0, n+1\leq j \leq s+1$, we get $|\hat \sigma_j|\leq ||\Delta||_2< \frac{4^{s+1}\sigma}{3}, n+1\leq j \leq s+1$. This proves (\ref{equ:numberdetectalgothm1equ0}). 

Let $\theta_{\min}(\Omega,s)=\frac{\Omega}{2s}\min_{p\neq q}\bonenorm{\vect y_p-\vect y_q}$. The separation condition (\ref{equ:numberdetectalgothm1equ1}) implies that
$$\theta_{\min}(\Omega,s) \geq  \frac{2(1+\sqrt{3})\pi}{3} \Big(\frac{2n4^{s+1}}{3}\frac{\sigma}{m_{\min}}\Big)^{\frac{1}{2n-2}}.$$
By Lemma \ref{lem:numberdetectalgolem0}, we have 
\begin{align}\label{MUSICthm1equ7}
\sigma_n \geq\frac{m_{\min}\big(3\theta_{\min}(\Omega,s)\big)^{2n-2}}{n(2(1+\sqrt{3})\pi)^{2n-2}}>2\frac{4^{s+1}\sigma}{3}.
\end{align}
Similarly, by Weyl's theorem, $|\hat \sigma_n-\sigma_n|\leq ||\Delta||_2$. Thus, $\hat \sigma_n\geq  2(s+1)\sigma-||\Delta||_2\geq \frac{4^{s+1}\sigma}{3}$. The conclusion (\ref{MUSICthm1equ2}) then follows. 
\end{proof}

\subsection{Coordinate-combination-based sweeping singular-value-thresholding number detection algorithm}

Based on Theorem \ref{thm:numberdetectalgothm1}, we can propose a simple thresholding algorithm, \textbf{Algorithm 1}, for the number detection.

\begin{algorithm}[H]\label{algo:coordcombinnumberalgo}
	\caption{\textbf{Coordinate-combination-based singular-value-thresholding number detection algorithm}}
	\textbf{Input:} Number $s$; Noise level $\sigma$;\\
	\textbf{Input:} Measurement: $\mathbf{Y}(\vect \omega), \vect \omega \in [0,1,\cdots, \Omega]^2$;\\
	\textbf{Input:} Translation vector $\vect v$ in $\mathbb R^2$; \\
	1: Modify the measurement and get $\vect X(\vect \omega) = e^{i \vect v^\top \vect \omega}\mathbf Y(\vect \omega)$\;
	2: Let $r = \Omega \mod 2s$, formulate $\vect D(t) = \sum_{t_1+t_2 =t, 0\leq t_1, t_2\leq t} {t\choose t_1} \vect X(\vect \omega_{rt_1, rt_2}), \quad t=0, \cdots, 2s$\;
	3: Assemble the $(s+1)\times(s+1)$ Hankel matrix $\mathbf H(s)$ like (\ref{equ:hankelmatrix1}) from $\vect D(t)$'s, and
	compute the singular value of $\mathbf H(s)$ as $\hat \sigma_{1}, \cdots,\hat \sigma_{s+1}$ distributed in a decreasing manner\;
	4: Determine $n$ by $\hat \sigma_n\geq \frac{4^{s+1}\sigma}{3}$ and $\hat \sigma_{j}< \frac{4^{s+1}\sigma}{3}, j=n+1,\cdots, s+1$\;
	\textbf{Return:} $n$ 
\end{algorithm}
Note that for \textbf{Algorithm 1} to work, in addition to the smallness of the noise level $\sigma$, we also need the integer $s$  to be larger than the source number. However, a suitable $s$ is not easy to estimate and large $s$ may incur a deterioration of the resolution as indicated by (\ref{equ:numberdetectalgothm1equ1}). To remedy this issue, we propose a sweeping singular-value-thresholding number detection algorithm (\textbf{Algorithm 2}) below.	In short, we detect the number $n_{recover}$ by \textbf{Algorithm 1} for all $s$ from $2$ to $\lfloor \frac{\Omega-1}{2}\rfloor$, and choose the greatest one $n_{\max}$ as the number of point sources. When the detected $n_{recover}$ becomes smaller than $n_{\max}$ for a large number of iterations, we will stop the loop. The details are summarized in \textbf{Algorithm2} below. 

We remark that when $s=n$ and the point sources satisfy 
\begin{equation}\label{equ:numberalgorithmequ1}
\min_{p\neq q}\bonenorm{\vect y_p-\vect y_q} \geq \frac{Cn}{\Omega}\Big(\frac{\sigma}{m_{\min}}\Big)^{\frac{1}{2n-2}},
\end{equation}
for some constant $C$, then (\ref{equ:numberdetectalgothm1equ1}) is satisfied. Thus by Theorem \ref{thm:numberdetectalgothm1}, for a suitable choice of $\vect v$, \textbf{Algorithm 1} can exactly detect the number $n$ when $s=n$. As $s$ increases to values greater than $n$, (\ref{equ:numberdetectalgothm1equ0}) implies that the number detected by \textbf{Algorithm 1} will not exceed $n$. Therefore, the sweeping singular-value-thresholding algorithm (\textbf{Algorithm \ref{algo:coordcombinsweepnumberalgo}}) can detect the exact number $n$ when $\Omega$ is greater than $2n+1$ and the point sources are separated by the minimal separation distance we derived in Theorem \ref{thm:highdupperboundnumberlimit0}. This demonstrates the optimal performance of \textbf{Algorithm \ref{algo:coordcombinsweepnumberalgo}}. We also remark that the theoretical threshold derived in Theorem \ref{thm:numberdetectalgothm1} seems to be larger than the one that is needed. One can improve the algorithm by choosing smaller threshold. Deriving new estimates for the thresholds in different cases is another interesting problem. 

\begin{algorithm}[H]\label{algo:coordcombinsweepnumberalgo}
	\caption{\textbf{Coordinate-combination-based sweeping singular-value-thresholding number detection algorithm}}	
	\textbf{Input:} Noise level $\sigma$; Measurement: $\mathbf{Y}$; Translation vector $\vect v$;\\
	\textbf{Input:} $n_{max}=0$, $smax_{index}=2$\\
	\For{$s=2: \lfloor \frac{\Omega-1}{2}\rfloor$}{
		Input $s,\sigma, \mathbf{Y}, \vect v$ to \textbf{Algorithm 1}, save the output of \textbf{Algorithm 1} as $n_{recover}$\; 
		\If{$n_{recover}>n_{max}$}{$n_{max}=n_{recover}$\;
		$smax_{index} = s$\;}
		\If{$s\geq$ $smax_{index}+2$} {break\;}}
	\textbf{Return} $n_{max}$.
\end{algorithm}

\subsection{Phase transition and  performance of Algorithm \ref{algo:coordcombinsweepnumberalgo}}\label{section:numberphasetransition}
In this subsection, we conduct numerical experiments to demonstrate the phase transition phenomenon regarding the super-resolution factor (SRF) and the SNR using \textbf{Algorithm \ref{algo:coordcombinsweepnumberalgo}}. We consider recovering the number of three and four sources. We fix $\Omega =10$ and detect the source number from their noisy Fourier data at $[0,1,\cdots, \Omega]^2$. We consider sources in $[0, \frac{\pi}{2}]^2$ and the translation vector in \textbf{Algorithm \ref{algo:coordcombinnumberalgo}} is $\vect v = (0, \frac{\pi}{2})^\top$. The noise level is $\sigma$ and the minimum separation distance between sources is $D_{\min}$. We perform $10000$ random experiments (the randomness is in the choice of ($d_{\min}$,$\sigma$, $\vect y_j$, $a_j$)) and detect the source number by \textbf{Algorithm \ref{algo:coordcombinsweepnumberalgo}}. We record the number of each successful detection (source number is detected exactly) and failed detection. Figures \ref{fig:twodnumberphasetransition} shows the result for the successful and unsuccessfully recovery in the parameter space  $\log(SNR)$ versus $\log(SRF)$ . It is observed that there is a line with slope ($2n-2$) in the parameter space of $\log(SRF)$-$\log(SNR)$ above which the number detection is always successful. This phase transition phenomenon is exactly the one predicted by our theoretical results in Theorems \ref{thm:highdupperboundnumberlimit0} and \ref{thm:numberdetectalgothm1}.  It also illustrates the efficiency of \textbf{Algorithm \ref{algo:coordcombinsweepnumberalgo}} as it can resolve the source number correctly in the regime where the source separation distance is of the order of the computational resolution limit.

\begin{figure}[!h]
	\centering
	\begin{subfigure}[b]{0.48\textwidth}
		\centering
		\includegraphics[width=\textwidth]{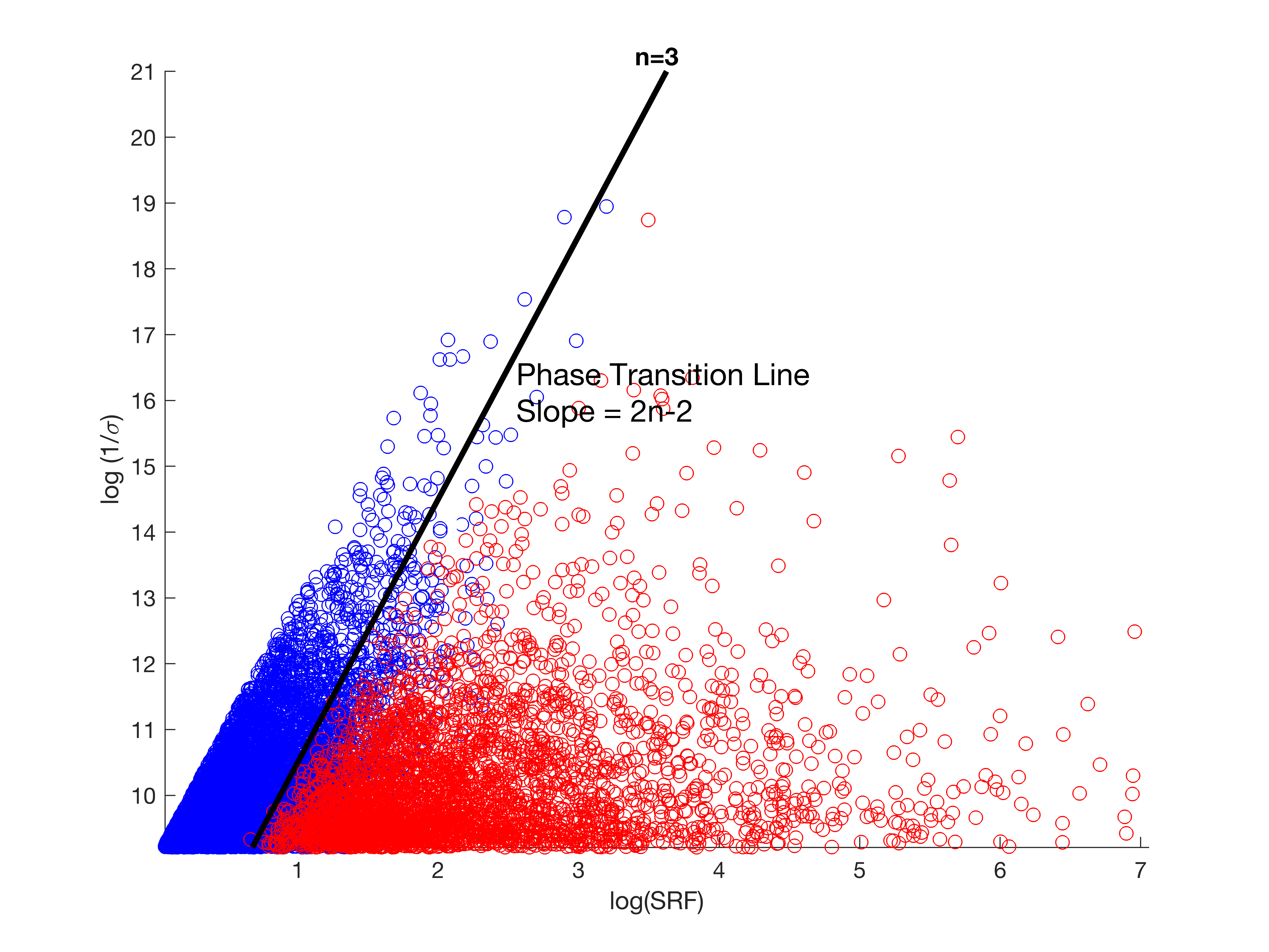}
		\caption{detection success}
	\end{subfigure}
	\begin{subfigure}[b]{0.48\textwidth}
		\centering
		\includegraphics[width=\textwidth]{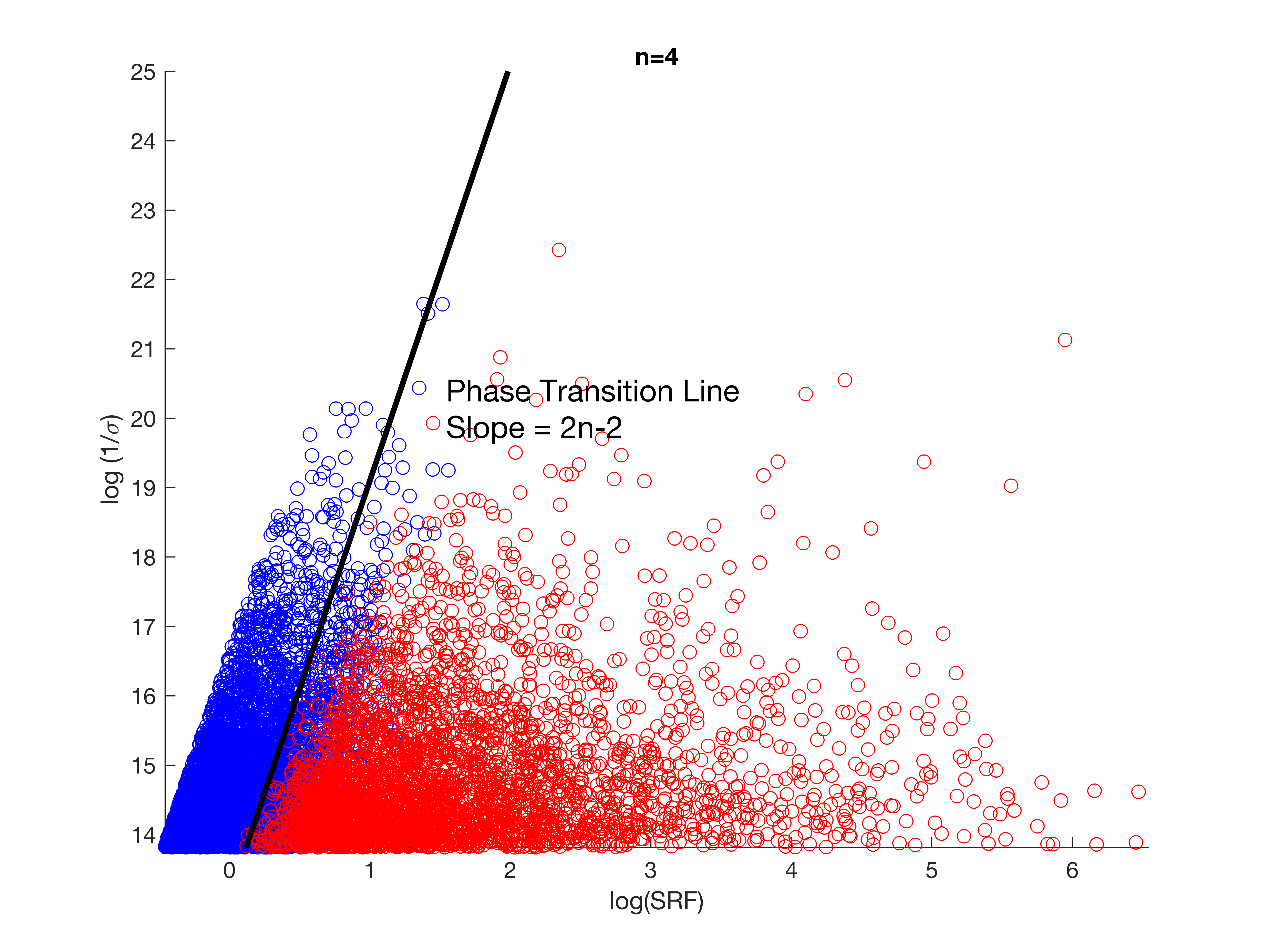}
		\caption{detection success}
	\end{subfigure}
	\caption{Plots of the successful and the unsuccessful number detection by \textbf{Algorithm \ref{algo:coordcombinsweepnumberalgo}} depending on the relation between $\log(SRF)$ and $\log(\frac{1}{\sigma})$. (a) illustrates that three sources can be exactly detected if $\log(\frac{1}{\sigma})$ is above a line of slope $4$ in the parameter space. (b) illustrates that four sources can be exactly detected if $\log(\frac{1}{\sigma})$ is above a line of slope $6$ in the parameter space.}
	\label{fig:twodnumberphasetransition}
\end{figure}


\section{An algorithm for the source reconstruction in two-dimensional DOA problems}\label{section:twodlocationrecoveryalgo}
In this section, based on the idea of coordinate-combination, we propose a new MUSIC algorithm for resolving the sources in the two-dimensional DOA estimation. Our algorithm is named as coordinate-combination-based MUSIC algorithm; see \textbf{Algorithm \ref{algo:coordcombinMUSICalgo}}.   

\subsection{Hankel matrix construction}\label{section:locationalgohankelconstruct}
Similarly to the number detection algorithm in the above section, the MUSIC algorithm also relies on a singular value decomposition of certain Hankel matrix. Compared to conventional MUSIC-based DOA algorithms, the main novelty of our algorithm lies in a different way of assembling Hankel matrices. Similarly to Section \ref{section:numberalgohankelconstruct}, the Hankel matrix construction here is also based on observations in Section \ref{section:threeobservations} and the details are presented below.

\medskip
\noindent\textbf{Measurement modification by source translation}\\
We consider the same model setting as (\ref{equ:DOAmeasuresetting1}) for the available measurement. We also perform the source translation and modify the measurement to get
\begin{equation}\label{equ:locationrecoveryalgomeasure1}
\begin{aligned}
    \vect X(\vect \omega) = & e^{i \vect v^\top \vect \omega} \vect Y(\vect \omega) = \sum_{j=1}^{n}a_j e^{i (\vect{y}_j+\vect v)^\top \vect{\omega}} + e^{i \vect v^\top \vect \omega} \mathbf W(\vect{\omega})\\
    = & \sum_{j=1}^{n}a_j e^{i \vect x_j^\top \vect{\omega}} +  \vect  {\tilde{W}}(\vect{\omega}), \quad  \vect{\omega}\in [0,1, \cdots, \Omega]^2,
\end{aligned}
\end{equation}
where $\vect x_j = \vect y_j +\vect v$ for a suitable $\vect v\in \mathbb R^2$ and $|\mathbf {\tilde{W}}(\vect{\omega})|< \sigma$. 

\medskip
\noindent \textbf{Measurement modification by the coordinate-combination technique}\\
Let $s=\lfloor \frac{\Omega}{2} \rfloor$. From the modified measurement $\vect X(\vect \omega)$, we construct the following two lists of data:
\begin{align*}
&\vect D(t) = \sum_{t_1+t_2 =t, 0\leq t_1, t_2\leq t} {t\choose t_1} \vect X(\vect \omega_{t_1, t_2}), \quad t=0, \cdots, 2s, \\
& \vect G(t) = \sum_{t_1+t_2 =t, 0\leq t_1, t_2\leq t} (-1)^{t_2}{t\choose t_1} \vect X(\vect \omega_{t_1, t_2}), \quad t=0, \cdots, 2s,
\end{align*}
where $\vect \omega_{t_1, t_2} = (t_1, t_2)^\top$. Note that 
\begin{align*}
\vect D(t) = & \sum_{j=1}^n a_j(e^{i \vect x_{j,1}}+ e^{i \vect x_{j,2}})^{t} + \sum_{t_1+t_2 =t, 0\leq t_1, t_2\leq t} {t\choose t_1}  \vect {\tilde{W}}(\vect \omega_{t_1, t_2})\\
= & \sum_{j=1}^na_j(e^{i \vect x_{j,1}}+ e^{i \vect x_{j,2}})^{t} + \vect {\hat W}_d(t), \\
\vect G(t) = & \sum_{j=1}^n a_j(e^{i \vect x_{j,1}}-e^{i \vect x_{j,2}})^{t} + \sum_{t_1+t_2 =t, 0\leq t_1, t_2\leq t} (-1)^{t_2}{t\choose t_1}  \vect {\tilde{W}}(\vect \omega_{t_1, t_2})\\
= & \sum_{j=1}^na_j(e^{i \vect x_{j,1}}- e^{i \vect x_{j,2}})^{t} + \vect {\hat W}_g(t), 
\end{align*}
where $\vect {\hat W}_d(t) = \sum_{t_1+t_2 =t, 0\leq t_1, t_2\leq t} {t\choose t_1}  \vect {\tilde{W}}(\vect \omega_{t_1, t_2})$ and $\vect {\hat W}_g(t) = \sum_{t_1+t_2 =t, 0\leq t_1, t_2\leq t} (-1)^{t_2}{t\choose t_1}  \vect {\tilde{W}}(\vect \omega_{t_1, t_2})$. 

\medskip
\noindent \textbf{Hankel matrix construction}\\
Finally, from these $\vect D(t), \vect G(t)$'s, we assemble the following Hankel matrices:
\begin{equation}\label{equ:hankelmatrix2}
\mathbf H_d(s)=\begin{pmatrix}
\mathbf D(0) &\mathbf D(1)&\cdots& \mathbf D(s)\\
\mathbf D(1)&\mathbf D(2)&\cdots&\mathbf D(s+1)\\
\cdots&\cdots&\ddots&\cdots\\
\mathbf D(s)&\mathbf D(s+1)&\cdots&\mathbf D(2s)
\end{pmatrix}, \quad
\mathbf H_g(s)=\begin{pmatrix}
\mathbf G(0) &\mathbf G(1)&\cdots& \mathbf G(s)\\
\mathbf G(1)&\mathbf G(2)&\cdots&\mathbf G(s+1)\\
\cdots&\cdots&\ddots&\cdots\\
\mathbf G(s)&\mathbf G(s+1)&\cdots&\mathbf G(2s)
\end{pmatrix}.
\end{equation}

\subsection{Standard MUSIC algorithm}\label{section:onedsupportalgorithm}
In this subsection, we perform the standard MUSIC algorithm \cite{schmidt1986multiple, stoica1989music, liao2016music, liu2022measurement} for the Hankel matrix $\vect H_d(s), \vect H_g(s)$ in (\ref{equ:hankelmatrix2}). For ease of presentation, we only introduce the MUSIC algorithm for $\vect H_d(s)$.  The one for $\vect H_g(s)$ can be developed in the same manner. Our algorithm first performs the singular value decomposition of $\vect H_d(s)$,
\[
\vect H_d(s) = \hat U\hat \Sigma \hat U^*=[\hat U_1\quad \hat U_2]\text{diag}(\hat \sigma_1, \hat \sigma_2,\cdots,\hat \sigma_n,\hat \sigma_{n+1},\cdots,\hat \sigma_{s+1})[\hat U_1\quad \hat U_2]^*,
\]
where $\hat U_1=(\hat U(1),\cdots,\hat U(n)), \hat U_2=(\hat U(n+1),\cdots,\hat U(s+1))$ with $n$ being the estimated source number (model order). The source number $n$ can be detected by \textbf{Algorithm \ref{algo:coordcombinsweepnumberalgo}} and many other algorithms such as those in \cite{akaike1998information, wax1985detection, schwarz1978estimating, wax1989detection, chen1991detection, he2010detecting, han2013improved, liu2021theorylse, liu2021mathematicalhighd}. Denote the orthogonal projection onto the space $\hat U_2$ by $\hat P_2x=\hat U_2(\hat U_2^*x)$. For a test vector $\Phi(d)=(1, d,\cdots,d^s)^\top$, one defines the MUSIC imaging functional 
\begin{align*}
\hat J(d)=\frac{||\Phi(d)||_2}{||\hat P_2\Phi(d)||_2}=\frac{||\Phi(d)||_2}{||\hat U_2^*\Phi(d)||_2}.
\end{align*}
The local  maximizer of $\hat J(d)$ indicates the supports of the sources. In practice, one can test evenly spaced points in a specified region and plot the discrete imaging functional and then determine the sources by detecting the peaks. In our case, we only need to test some discrete points $d\in \mathbb C$ with $|d|\leq 2$ and select the peak by certain algorithms (such as the one in \cite{liu2022measurement} or its two-dimensional analog). Finally, we summarize the standard MUSIC algorithm in \textbf{Algorithm \ref{algo:standardmusic}} below.

\begin{algorithm}[H]
	\caption{\textbf{Standard MUSIC algorithm}}
	\textbf{Input:} Source number $n$\;
	\textbf{Input:} Modified measurements: $\vect D(t)$ (or $\vect G(t)$), $t=0, \cdots, s$ with $s\geq n$\;
	\textbf{Input:} Test points $d$'s\;
	1: Formulate the $(s+1)\times (s+1)$ Hankel matrix $\vect H_d(s)$ from $\vect D(t)$'s as (\ref{equ:hankelmatrix2})\;
	2: Compute the singular vectors of $\vect H_d(s)$ as $\hat U(1), \hat U(2),\cdots,\hat U(s+1)$ and form the noise space $\hat U_{2}=(\hat U(n+1),\cdots,\hat U(s +1))$\;
	3: For test points $d$'s, construct the test vector $\Phi(d)=(1,d, \cdots, d^s)^\top$\; 
	4: Plot the MUSIC imaging functional $\hat J(d)=\frac{||\Phi(d)||_2}{||\hat U_2^*\Phi(d)||_2}$\;
	5: Select the peak locations $\hat d_j$'s in the plot of $\hat J(d)$.
	\label{algo:standardmusic}
\end{algorithm}

\subsection{Coordinate-combination-based MUSIC algorithm}
After applying the MUSIC algorithm to both $\vect H_d(s), \vect H_g(s)$, we expect to reconstruct $n$ $\hat d_j$'s which is close to $d_j = e^{i \vect x_{j,1}}+ e^{i \vect x_{j,2}}$, and $n$ $\hat g_j$'s which is close to $g_j = e^{i \vect x_{j,1}}- e^{i \vect x_{j,2}}$. The next question is how to link the pair $\hat d_j, \hat g_j$ that correspond to the same source. This is an inevitable pair matching issues in most of the two-dimensional DOA algorithms \cite{liu2021mathematicalhighd}, where ad hoc schemes \cite{zoltowski1989sensor, johnson1991operational, chen1992direction, yilmazer2006matrix} were derived to associate the estimated azimuth and elevation angles. Here, in contrast with  conventional DOA algorithms, we do not need to link the azimuth and elevation angles but to link  $\hat d_j$ and $\hat g_j$. 

Observe that $|d_j+g_j| = |2 e^{i \vect x_{j,1}}|=2$ and $|d_j-g_j| = |2 e^{i \vect x_{j,2}}|=2$. We can use this criterion to match the pair $\hat d_j, \hat g_j$ that they should satisfy 
\begin{equation}\label{equ:pairmatchingcriterion1}
|\hat d_j + \hat g_j|\approx 2, \quad |\hat d_j - \hat g_j|\approx 2.
\end{equation}
For example, we could consider the following minimization problem:
\begin{equation}\label{equ:CBMUSICequ1}
\min_{\pi \in \zeta(n)} \sum_{j=1}^n\babs{|\hat d_j + \hat g_{\pi_j}|-2}+\babs{|\hat d_j - \hat g_{\pi_j}|-2},
\end{equation}
where $\zeta(n)$ is the set of all permutations of $\{1, \cdots, n\}$. This can be viewed as a balanced assignment problem \cite{pentico2007assignment}, which can be solved efficiently by many algorithms such as the Hungarian algorithm.  

We remark that our pair matching algorithm is not the one usually required in other one-dimensional based DOA algorithms. Unlike our case, the other pair matching problem is not an assignment problem, wherefore the pair matching is usually time consuming or complex processing is conducted to reduce the computational cost.

\begin{algorithm}[H]\label{algo:coordcombinMUSICalgo}
	\caption{\textbf{Coordinate-combination-based MUSIC algorithm for two-dimensional DOA}}
	\textbf{Input:} Source number $n$; noise level $\sigma$;\\
	\textbf{Input:} Measurement: $\mathbf{Y}(\vect \omega), \vect \omega \in [0,1,\cdots, \Omega]^2$;\\
	\textbf{Input:} Translation vector $\vect v$ in $\mathbb R^2$; \\
	\textbf{Input:} Evenly spaced test points $d\in \mathbb{C}$ with $|d|\leq 2$;\\
	1: Modify the measurement and get $\vect X(\vect \omega) = e^{i \vect v^\top \vect \omega}\mathbf Y(\vect \omega)$\;
	2: Let $s = \lfloor \frac{\Omega}{2}\rfloor$, formulate $\vect D(t) = \sum_{t_1+t_2 =t, 0\leq t_1, t_2\leq t} {t\choose t_1} \vect X(\vect \omega_{t_1, t_2}),\quad \vect G(t) = \sum_{t_1+t_2 =t, 0\leq t_1, t_2\leq t}(-1)^{t_2} {t\choose t_1} \vect X(\vect \omega_{t_1, t_2}), \quad t=0, \cdots, 2s$\;
	3: Input $\vect D, n$ and test points $d$'s into \textbf{Algorithm \ref{algo:standardmusic}} and get the output $\hat d_1, \cdots, \hat d_n$\;
	4: Input $\vect G, n$ and test points $d$'s into \textbf{Algorithm \ref{algo:standardmusic}} and get the output $\hat g_1, \cdots, \hat g_n$\;
	5: Matching the $\hat d_j , \hat g_j$'s by applying an assignment algorithm (match pairs in matlab) to solve (\ref{equ:CBMUSICequ1}) and get the pair list $\{(\hat d_j, \hat g_j)\}^{j=1^n}$\;
	6: Get $\frac{\hat d_j +\hat g_j}{2}$ and  $\frac{\hat d_j -\hat g_j}{2}, j=1, \cdots, n$. Get $e^{i \vect {\hat x}_{j,1}r}$ by considering the closest point to $\frac{\hat d_j +\hat g_j}{2}$ on the unit circle. Get $e^{i \vect {\hat x}_{j,2}}$ by considering the closest point to $\frac{\hat d_j -\hat g_j}{2}$ on the unit circle\;
	7: The recovered $\vect {\hat x}_j = (\vect {\hat x}_{j,1}, \vect {\hat x}_{j,2})^\top$. Reconstruct $\vect {\hat y}_j = \vect {\hat x}_j-\vect v, j=1, \cdots, n$\;
	\textbf{Return:} $\vect {\hat y}_1, \cdots, \vect {\hat y}_n$.  
\end{algorithm}

\subsection{Superiority of the algorithm}\label{section:superioryofalgo}
\subsubsection{Overcome the issue of separation distance loss in conventional two-dimensional DOA algorithms}
Despite the fact that different recovering methods are proposed for DOA estimation in two dimensions, the conventional way for tackling the problem has hardly exceeded the scope of recovering the two direction (x- and  y-direction) components of sources individually. Thus, as illustrated in Figure \ref{fig:severe loss of distance}, severe loss of the source separation distance in one dimension is always an inevitable issue that causes unstable recovery of the direction components. Most of the researches ignored this issue and some papers \cite{wang2008tree, wang2015decoupled} proposed ad hoc schemes to enhance the reconstruction but in a complex manner. 

Our method is a new one-dimensional-based algorithm where the issue of severe source separation distance loss is avoided in a simple way. In our algorithm,  the separation distance between direction components of sources are still preserved. This has been demonstrated by Lemma \ref{lem:projectiondislem1} for $\vect \theta_j \in \mathbb [0, \pi]^2, j=1,2$ with $\frac{\pi}{3}\leq  \vect \theta_{j,2} - \vect \theta_{j,1}\leq \frac{2}{3}\pi, j=1,2$. Furthermore,  Theorem \ref{thm:dislossthm1} shows that, for $\vect y_j \in [0, \frac{\pi}{2}]^2$ and $\vect v = (0, \frac{\pi}{2})^\top$, the separation distance between $\vect x_j = \vect y_j+\vect v$'s can be preserved after the coordinate-combination. By Theorem \ref{thm:dislossthm1}, if the distance between the $\vect x_j$'s is a certain constant $C$, then the distance between $e^{i \vect x_{j,1}}+ e^{i \vect x_{j,2}}$ is larger than $\frac{2C}{\pi^2}$ times the original distance. For better results of preservation of the distance, as indicated by Theorems \ref{thm:highdupperboundnumberlimit0} and \ref{thm:highdupperboundsupportlimit0}, we could consider sources in a smaller region with a specified translation. In the numerical experiments presented in this paper, for ease of discussion and presentation, we will consider sources in $[0,\frac{\pi}{2}]^2$ and the translation vector $\vect v = (0, \frac{\pi}{2})^\top$. We leave the recovering strategies of the whole region $[0,2\pi]^2$ and other enhancement for future works.  

\begin{thm}\label{thm:dislossthm1}
For two different vectors $\vect x_j \in \mathbb [0, \frac{\pi}{2}]\times[\frac{\pi}{2}, \pi], j=1,2$, if $||\vect x_1 - \vect x_2||_1\geq C$ for a constant $C$, then 
	\[
	\babs{e^{i \vect x_{1,1}}+ e^{i \vect x_{1,2}}- (e^{i \vect x_{2,1}} + e^{i \vect x_{2,2}})}\geq \frac{2C}{\pi^2} C.
	\]
\end{thm}
\begin{proof}
We prove the lemma by considering the following two cases. \\
\textbf{Case 1:} $0\leq \vect x_{1,1}  \leq \vect x_{2,1} \leq  \vect x_{2,2}\leq  \vect x_{1,2}\leq \pi$.\\
In this case, 
\begin{align*}
\babs{e^{i \vect x_{1,1}}+ e^{i \vect x_{1,2}}- (e^{i \vect x_{2,1}} + e^{i \vect x_{2,2}})} \geq & \babs{e^{i \vect x_{2,1}}+ e^{i\vect x_{2,2}}}- \babs{e^{i \vect x_{1,1}} + e^{i\vect x_{1,2}}}\\
	\geq &2\Big(\cos(\frac{\phi_2}{2}) - \cos(\frac{\phi_1}{2})\Big),
\end{align*}
where $\phi_j = \vect x_{j,2} - \vect x_{j,1}, j=1,2$. By the assumption of the theorem, we have $C\leq \phi_1 -\phi_2 \leq \pi$ and $C\leq \phi_1+\phi_2\leq 2\pi$. Thus 
\[
2\Big(\cos(\frac{\phi_2}{2}) - \cos(\frac{\phi_1}{2}) \Big) =4\sin(\frac{\phi_1+\phi_2}{4})\sin(\frac{\phi_1-\phi_2}{4})\geq 4 \sin(\frac{C}{4})\sin(\frac{C}{4})\geq  \frac{2C^2}{\pi^2}.
\]
where the last inequality uses $\sin(\frac{C}{4}) \geq \frac{2\sqrt{2}}{\pi}\frac{C}{4}$ for $0\leq \frac{C}{4}\leq \frac{\pi}{4}$. \\
\textbf{Case 2:} $0\leq \vect x_{1,1}  \leq  \vect x_{2,1} \leq   \vect x_{1,2}\leq  \vect x_{2,2}\leq \pi$.\\
Again, the idea is to calculate the angle between $e^{i \vect x_{1,1}}+ e^{i \vect x_{1,2}}$ and $e^{i \vect x_{2,1}} + e^{i \vect x_{2,2}}$. By a simple analysis of the angle relations between $e^{i \vect x_{1,1}}, e^{i \vect x_{1,2}}, e^{i \vect x_{2,1}}, e^{i \vect x_{2,2}},$ we obtain that the angle between $e^{i \vect x_{1,1}}+ e^{i \vect x_{1,2}}$ and $e^{i \vect x_{2,1}} + e^{i \vect x_{2,2}}$ is $\frac{\vect x_{2,1}-\vect x_{1,1}+ \vect x_{2,2}-\vect x_{1,2}}{2}$ which is larger than $\frac{C}{2}$. Thus 
\begin{equation}\label{equ:dislossequ1}
\babs{e^{i \vect x_{1,1}}+ e^{i \vect x_{1,2}}- (e^{i \vect x_{2,1}} + e^{i \vect x_{2,2}})} \geq \max\Big(\babs{e^{i \vect x_{1,1}}+ e^{i \vect x_{1,2}}}, \babs{e^{i \vect x_{2,1}} + e^{i \vect x_{2,2}}}\Big)\sin(\frac{C}{2}).
\end{equation}
We next claim that 
\[
\max\Big(\babs{e^{i \vect x_{1,1}}+ e^{i \vect x_{1,2}}}, \babs{e^{i \vect x_{2,1}} + e^{i \vect x_{2,2}}}\Big) \geq 2 \cos\Big(\frac{\pi-C/2}{2}\Big).
\]
Otherwise, $\vect x_{1,2}-\vect x_{1,1}> \pi- \frac{C}{2}$ and $\vect x_{2,2}-\vect x_{2,1}> \pi- \frac{C}{2}$, which is impossible when $||\vect x_1 -\vect x_2||_1\geq C$. Thus the claim is proved. Together with (\ref{equ:dislossequ1}), we arrive at 
\[
\babs{e^{i \vect x_{1,1}}+ e^{i \vect x_{1,2}}- (e^{i \vect x_{2,1}} + e^{i \vect x_{2,2}})}\geq 2 \sin(\frac{C}{4})\sin(\frac{C}{2}) \geq \frac{2C^2}{\pi^2}.
\]
This completes the proof.
\end{proof}

\begin{figure}[!h]
		\centering
		\includegraphics[width=0.8\textwidth]{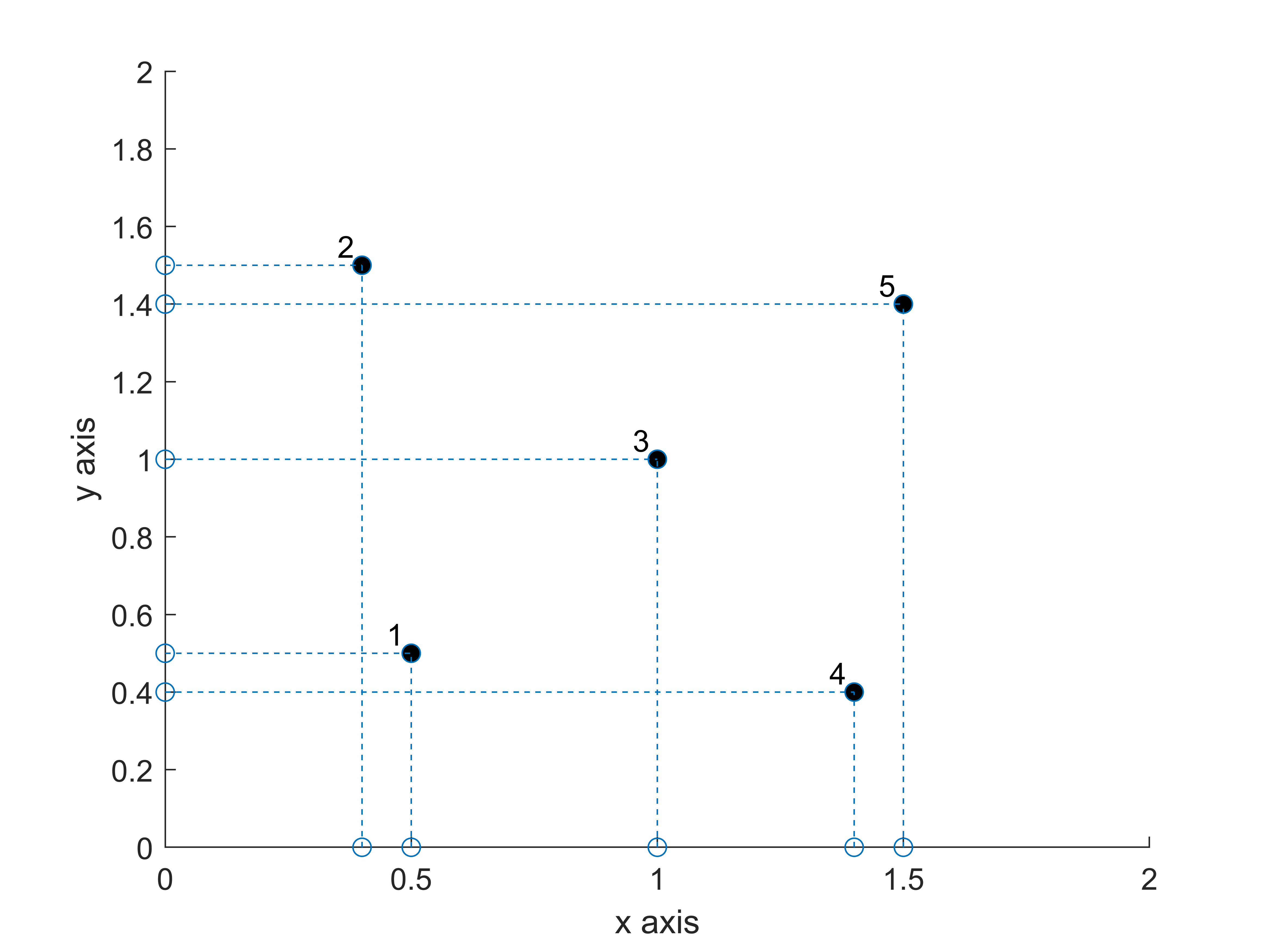}
	\caption{Although the  sources are well-separated, the direction components of sources are closely spaced.}
	\label{fig:severe loss of distance}
\end{figure}

\subsubsection{Phase transition and performance of Algorithm \ref{algo:coordcombinMUSICalgo}}\label{section:locationphasetransition}
Most of the conventional two-dimensional DOA algorithms consider multiple snapshots of measurements from coherent or incoherent signals. Also, the noise is usually assumed to be white Gaussian noise such that the expectation of the covariance matrix of the measurement vector is a sum of two terms, where  the first term is from the correlation of the signals and the second one is the noise correlation matrix. Based on this crucial observation, many algorithms were derived to tackle the problem. Differently to the above model, we consider recovering the source from a single measurement with deterministic noise. Thus we do not compare the performance of our algorithm with those algorithms with statistical model. We demonstrate the super-resolution capacity of our algorithm for the single snapshot case by showing the phase transition of the algorithm. We will derive a coordinate-combination-based MUSIC algorithm for multiple snapshots case in a forthcoming work. 

We now describe the numerical experiments for demonstrating the phase transition phenomenon of our algorithm in terms of the SNR versus the super-resolution factor. We fix $\Omega=10$ and  consider three and four sources separated by the minimum separation distance $D_{\min}$, i.e., $\min_{p\neq q}||\vect y_p - \vect y_q||_1\geq D_{\min}$. We perform 10000 random experiments (the randomness is in the choice of $(D_{\min},\sigma, \vect y_j, a_j)$ to recover the sources using \textbf{Algorithm \ref{algo:coordcombinMUSICalgo}}. The reconstruction is viewed and recorded as successful if the recovered source is in a $\frac{D_{\min}}{3}$-neighborhood of the underlying source, otherwise it is unsuccessful; See \textbf{Algorithm \ref{algo:singleexperiemnt}} for the details of a single experiment. The results of the experiments are summarized in Figure \ref{fig:twodnumberphasetransition} which shows each successful and unsuccessfully recovery with respective to the $\log(SRF)$ and $\log(SNR)$. It is observed that there is a line with slope ($2n-1$) in the parameter space  $\log(SRF)$ versus $\log(SNR)$ above which the source is stably reconstructed for every realization. This phase transition phenomenon is exactly the one predicted by our theoretical result in Theorems \ref{thm:highdupperboundsupportlimit0}.  It also manifests the efficiency of \textbf{Algorithm \ref{algo:coordcombinMUSICalgo}} as it can resolve the source in the regime where the source separation distance is of the order of the computational resolution limit. 

\begin{algorithm*}[H]\label{algo:singleexperiemnt}
	\caption{\textbf{A single experiment}}	
	\textbf{Input:} Sources $\mu=\sum_{j=1}^{n}a_j \delta_{\vect y_j}$; Noise level $\sigma$;\\
	\textbf{Input:} Measurements: $\mathbf{Y}(\vect \omega), \vect \omega = [0,1,\cdots, \Omega]^2$;\\
	1: $\text{Successnumber}=0$\;
	2: Input source number $n$ and measurement $\vect Y$ to \textbf{Algorithm \ref{algo:coordcombinMUSICalgo}} and save the output as $\vect y_1, \cdots, \vect y_n$\; 
	\For{each $1\leq j \leq n$}{
		Compute the error for the source location $\vect y_j$:
		$e_j:=\min_{\mathbf{\hat y}_l, l=1, \cdots,n }||\mathbf{\hat y}_l- \vect y_j||_2$\; 
		The source location $\vect y_j$ is recovered successfully if
		\[e_j< \frac{\min_{p\neq j}||\vect y_p- \vect y_j||_2}{3};\]
		and 
		\[\text{Successnumber}=\text{Successnumber}+1;\] 
	}  
	\eIf{$\text{Successnumber}==n$}{
		Return Success}
	{Return Fail}.
\end{algorithm*}

\begin{figure}[!h]
	\centering
	\begin{subfigure}[b]{0.48\textwidth}
		\centering
		\includegraphics[width=\textwidth]{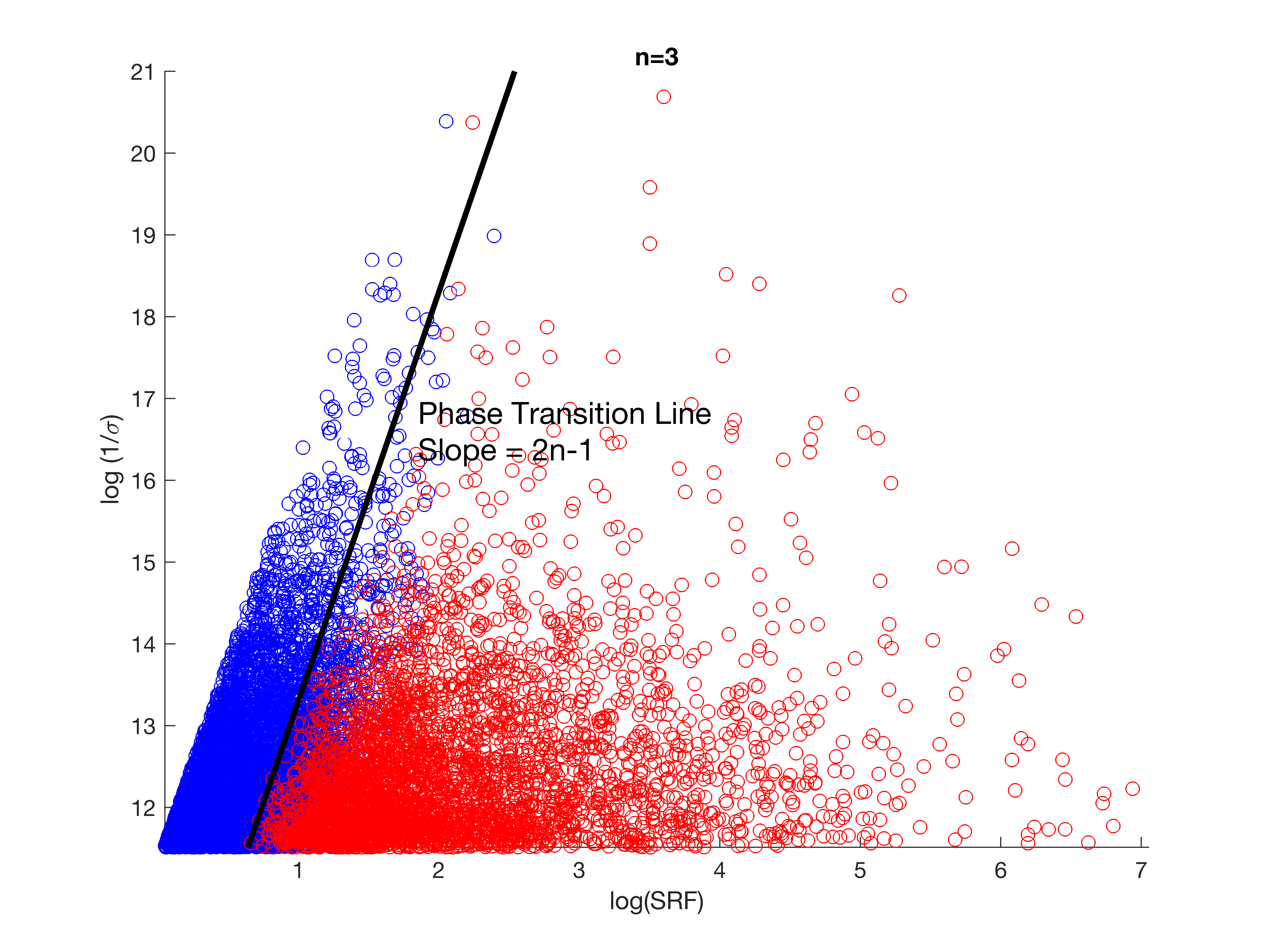}
		\caption{Recovery success.}
	\end{subfigure}
	\begin{subfigure}[b]{0.48\textwidth}
		\centering
		\includegraphics[width=\textwidth]{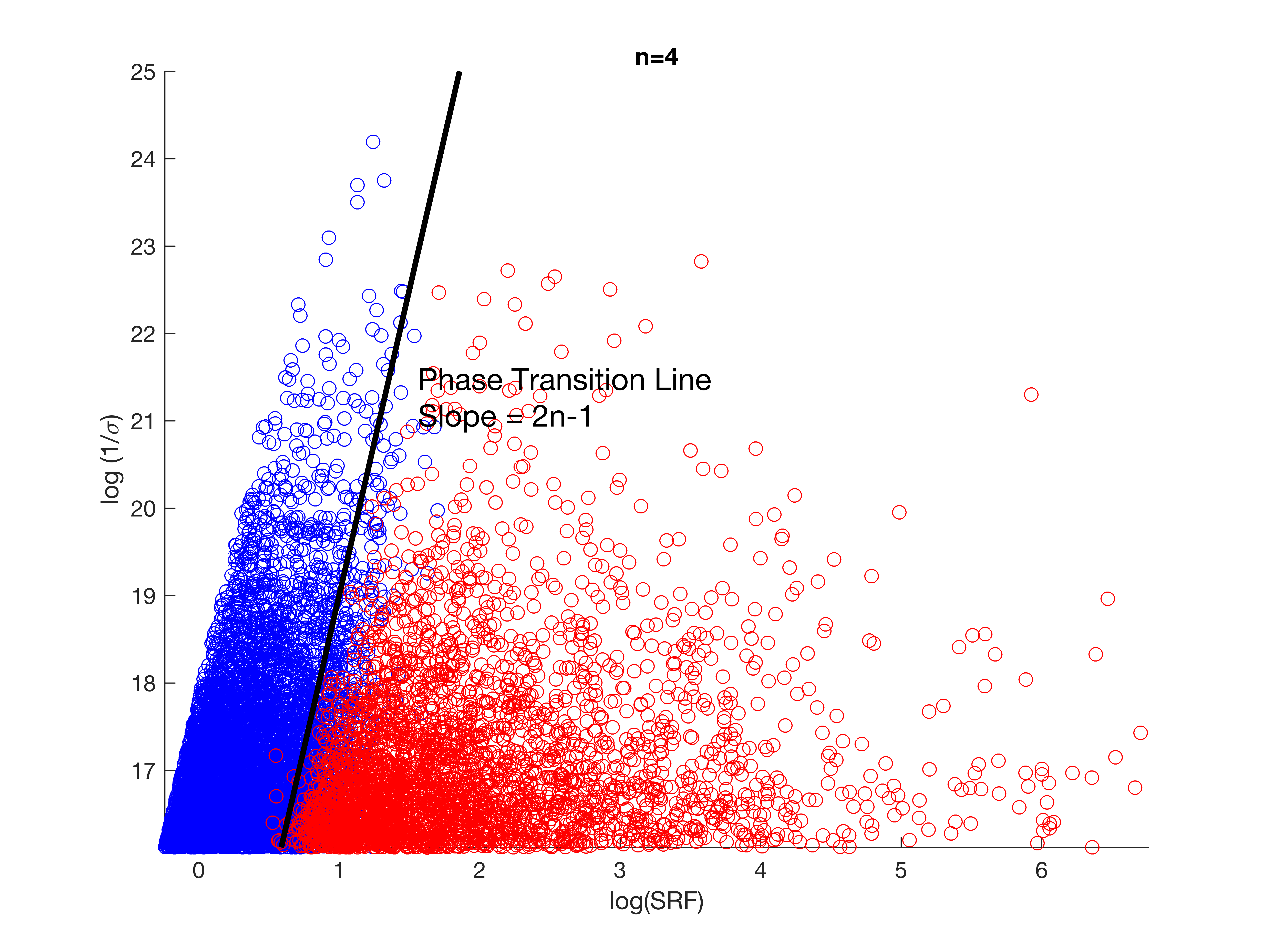}
		\caption{Recovery success.}
	\end{subfigure}

	\caption{Plots of the successful and the unsuccessful location recoveries by \textbf{Algorithm \ref{algo:coordcombinMUSICalgo}} in terms of $\log(\frac{1}{\sigma})$ versus $\log(SRF)$. (a) illustrates that locations of three point sources can be stably recovered if $\log(\frac{1}{\sigma})$ is above a line of slope $5$ in the parameter space. Conversely, for the same case, (b) shows that  locations of four point sources can be stably recovered if $\log(\frac{1}{\sigma})$ is above a line of slope $7$ in the parameter space.}
	\label{fig:twodsupportphasetransition}
\end{figure}

\section{A nonlinear approximation theory in Vandermonde space}\label{section:approxtheoryinvanderspace} 
In this section, we introduce the main technique, a nonlinear approximation theory in Vandermonde space\cite{liu2021mathematicaloned, liu2021theorylse}, 
that is used to deal with one-dimensional super-resolution problems. In \cite{liu2021mathematicaloned}, we have derived the theory for real numbers and in \cite{liu2021theorylse} for complex numbers on the unit circle. Here, we derive a different theory for arbitrary bounded complex numbers, which are related to the proofs of the main results of the paper. 

For a given positive integer $s$ and $\omega\in \mathbb C$, we denote by
\begin{equation}\label{equ:defineofphi}
	\phi_s(\omega)=(1,\omega,\cdots,\omega^{s})^\top 
\end{equation}
and call $\phi_s$ a Vandermonde vector. At the heart of the theory is the following nonlinear approximation problem in the Vandermonde space 
\begin{equation}\label{equ:vandermondeapprox}
	\min_{\hat a_j, \hat d_j\in \mathbb{R},|\hat d_j|\leq d,j=1,\cdots,k}\Big|\Big|\sum_{j=1}^k \hat a_j\phi_s(\hat d_j)-v\Big|\Big|_2,
\end{equation}
where $v=\sum_{j=1}^{k+1}a_j\phi_s(d_j)$ is a given vector. 
We shall derive a sharp lower-bound for this problem. In addition, we shall also investigate the stability of the approximation problem (\ref{equ:vandermondeapprox}) for $v=\sum_{j=1}^{k}a_j\phi_s(d_j)$.

\subsection{Notation and Preliminaries} \label{sec-vand-pre}
We introduce some notation and preliminaries. We denote the Vandermonde matrix by
\begin{align} \label{equ:vandermatrix}
	V_{s}(k)=\begin{pmatrix}
		1&\cdots&1\\
		d_1&\cdots&d_{k}\\
		\vdots&\ddots&\vdots\\
		d_1^s&\cdots&d_{k}^{s}
	\end{pmatrix}=
	\Big(
	\phi_s(d_1)\ \ \phi_s(d_2)\ \ \cdots\ \ \phi_s(d_k) 	
	\Big).
\end{align}
For a real matrix or a vector $A$, we denote by $A^\top$ its transpose and by $A^*$ its conjugate transpose. 

\medskip
We first present some basic properties of Vandermonde matrices.
\begin{lem}{\label{lem:norminversevandermonde0}}
	For $k$ distinct complex numbers $d_j$'s, we have
	\[
	||V_{k-1}(k)^{-1}||_{\infty}\leq \max_{1\leq i\leq k}\Pi_{1\leq p\leq k,p\neq i}\frac{1+|d_p|}{|d_i-d_p|},
	\]
	where $V_{k-1}(k)$ is the Vandermonde matrix $V_{k-1}(k)$ defined as in (\ref{equ:vandermatrix}).
\end{lem}
\begin{proof} See Theorem 1 in \cite{gautschi1962inverses}.\end{proof}

\vspace{0.2cm}
As a consequence, we directly have the following corollary.  
\begin{cor}\label{lem:norminversevandermonde1}
	Let $d_{\min}=\min_{i\neq j}|d_i-d_j|$ and assume that $\max_{i=1,\cdots,k}|d_i|\leq d$. Then 
	\[
	||V_{k-1}(k)^{-1}||_{\infty}\leq \frac{(1+d)^{k-1}}{(d_{\min})^{k-1}}.
	\] 
\end{cor}

\vspace{0.2cm}
\begin{lem}\label{lem:singularvaluevandermonde2}
	For distinct $d_1,\cdots, d_k \in \mathbb C$, define the Vandermonde matrices $V_{k-1}(k), V_s(k)$ as in (\ref{equ:vandermatrix}) with $s\geq k-1$. Then the following estimate on their singular values holds:
	\[
	\frac{1}{\sqrt{k}} \frac{1}{||V_{k-1}(k)^{-1}||_{\infty}} \leq \frac{1}{||V_{k-1}(k)^{-1}||_{2}}\leq  \sigma_{\min}(V_{k-1}(k))\leq \sigma_{\min}(V_{s}(k)).
	\]
\end{lem}
\begin{proof}
The result holds by using properties of matrix norms.
\end{proof}

Denote by
\[
S_{1k}^j:=\Big\{\{\tau_1,\cdots,\tau_j\}: \text{$\tau_p\in \{1,\cdots,k\},p=1,\cdots,j$ and $\tau_p\neq \tau_q$, for $p\neq q$}\Big\}.
\]
Note that there is no order in $\{\tau_1,\cdots,\tau_j\}$, i.e., $\{1,2\}$ and $\{2,1\}$ are the same sets. We then have the following decomposition of the Vandermonde matrix. 

\begin{prop}\label{vandermondegaussianelimiate1}
The Vandermonde matrix $V_{k}(k)$ defined as in (\ref{equ:vandermatrix}) can be reduced to the following form by using elementary column-addition operations, i.e.,
\begin{align}\label{equ:vandermondegaussianelimiate1}
V_{k}(k)G(1)\cdots G(k-1)DQ(1)\cdots Q(k-1)=
\begin{pmatrix}
1&0&\cdots&0\\
0&1&\cdots&0\\
\vdots&\vdots&\ddots&\vdots\\
0&0&\cdots&1\\
v_{(k+1)1}&v_{(k+1)2}&\cdots&v_{(k+1)k}
\end{pmatrix},
\end{align}
where $G(1),\cdots,G(k-1),Q(1),\cdots,Q(k-1)$ are elementary column-addition matrices, $$D=\text{diag}(1,\frac{1}{(d_2-d_1)},\cdots,\frac{1}{\Pi_{p=1}^{k-1}(d_k-d_p)})$$ and
\begin{equation}\label{equ:vandermondegaussianelimiate2}
v_{(k+1)j}=(-1)^{k-j}\sum_{\{\tau_1,\cdots,\tau_{k+1-j}\}\in S_{1k}^{k+1-j}}d_{\tau_1}\cdots d_{\tau_{k+1-j}}.
\end{equation}

\end{prop}	
\begin{proof}
    See Appendix B in \cite{liu2021mathematicaloned}. 
\end{proof}

\begin{lem}\label{lem:projectvolumeratiolem1}
	For an $s\times k$ complex matrix $A$ of rank $k$ with $s>k$, let $V$ be the space spanned by columns of $A$ and $V^{\perp}$ be the orthogonal complement of $V$. Denote by $P_{V^{\perp}}$ the orthogonal projection to $V^{\perp}$, and set $D=(A,v)$. We have
	\[
	\min_{a\in \mathbb C^{k}}||Aa-v||_2=||P_{V^{\perp}}(v)||_2=	
	\sqrt{\frac{\det(D^*D)}{\det(A^*A)}}.
	\]
\end{lem}
\begin{proof}
	See Lemma 1 in \cite{liu2021theorylse}.
\end{proof}

\begin{lem}\label{vandemondevolumeratio2}
	We have 
	\begin{equation} \label{eq-deter}
		\sqrt{\frac{\det(V_{k}(k)^{*}V_k(k))}{\det(V_{k-1}(k)^{*}V_{k-1}(k))}}= \sqrt{\sum_{j=0}^{k}|v_{j}|^2},
	\end{equation}
	where $V_{s}(k)$ is defined as in (\ref{equ:vandermatrix}) and $v_{j}=\sum_{\{\tau_1,\cdots,\tau_j\}\in S_{1k}^j}d_{\tau_1}\cdots d_{\tau_j}$. Especially, if $|d_j|<d,j=1,\cdots,k$, then
	\begin{equation}\label{equ:projectiondisestimate1}
		\sqrt{\frac{\det(V_{k}(k)^{*}V_k(k))}{\det(V_{k-1}(k)^{*}V_{k-1}(k))}}\leq (1+d)^k.
	\end{equation}	
\end{lem}
\begin{proof} Note that in Proposition \ref{vandermondegaussianelimiate1}, 
all the elementary column-addition matrices have unit determinant. As a result,
$ \det(V_{k}(k)^{*}V_k(k)) = \frac{ \det(F^{*}F)}{\det( D^*D)}$, where $F$ is the matrix in the right-hand side of (\ref{equ:vandermondegaussianelimiate1}), and $D$ is the diagonal matrix in Proposition \ref{vandermondegaussianelimiate1}. A direct calculation shows that $\det(F^{*}F)= \sum_{j=0}^{k}|v_j|^2$, where we use (\ref{equ:vandermondegaussianelimiate2}). On the other hand, $V_{k-1}(k)$ is a standard Vandermonde matrix and we have $\det(V_{k-1}(k)^*V_{k-1}(k)) =\frac{1}{\det (D^*D)}$. Combining these results, 
(\ref{eq-deter}) follows.  The last statement can be derived from (\ref{eq-deter}) and the estimate that 
\[
\sqrt{\sum_{j=0}^{k}|v_j|^2}\leq \sum_{j=0}^{k}|v_j|\leq \sum_{j=0}^{k}
\begin{pmatrix}
	k\\
	j
\end{pmatrix}
d^j=(1+d)^k.
\]
\end{proof}

\medskip
For reader's convenience, we finally present two auxiliary lemmas. For positive integers $p,q$ and complex numbers $z_1, \cdots, z_p, \hat z_1, \cdots, \hat z_q$, we define
\begin{equation}\label{equ:defineofeta}
	\eta_{p,q}(z_1, \cdots, z_p, \hat z_1, \cdots, \hat z_q)= 
	\begin{pmatrix}
		|z_1-\hat z_1|\cdots|z_1-\hat z_q|\\
		|z_2-\hat z_1|\cdots|z_2-\hat z_q|\\
		\vdots\\
		|z_p-\hat z_1|\cdots|z_p-\hat z_q|
	\end{pmatrix}.
\end{equation}
The following two properties of $\eta_{p,q}$ hold. 

\begin{lem}\label{lem:multiproductlowerbound0}
	For complex numbers $d_j, \hat d_j$'s, we have the following estimate
	\[
	\Big|\Big|\eta_{k+1,k}(d_1, \cdots, d_{k+1}, \hat d_1, \cdots, \hat d_k)\Big|\Big|_{\infty}\geq (\frac{d_{\min}}{2})^k,
	\]
	where $d_{\min}=\min_{j\neq p}|d_j-d_p|$ and $\eta_{k+1,k}(d_1, \cdots, d_{k+1}, \hat d_1, \cdots, \hat d_k)$ is defined as in (\ref{equ:defineofeta}).	
\end{lem}
\begin{proof}
Because we have $k+1$ $d_j$'s and only $k$ $\hat d_j$'s, there must exist one $d_{j_0}$ so that 
\[
|d_{j_0}-\hat d_j|\geq \frac{d_{\min}}{2}, \quad j =1, \cdots, k.
\]
Then the estimate in the lemma follows. 
\end{proof}

\begin{lem}\label{lem:multiproductstability1}
	Let $d_j, \hat d_j\in \mathbb C, j=1,\cdots, k$ satisfy $|d_j|, |\hat d_j|\leq d$. Assume that 
	\begin{equation}\label{equ:satblemultiproductlemma1equ1}
		||\eta_{k,k}(d_1, \cdots, d_k, \hat d_1, \cdots, \hat d_k)||_{\infty}< \epsilon, 
	\end{equation}
	where $\eta_{k,k}(\cdots)$ is defined as in (\ref{equ:defineofeta}), and that 
	\begin{equation}\label{equ:satblemultiproductlemma1equ2}
		d_{\min} = \min_{p\neq q}|d_p-d_j| \geq 2\epsilon^{\frac{1}{k}}.
	\end{equation}
	Then after reordering $d_j$'s, we have
	\begin{equation}\label{equ:satblemultiproductlemma1equ4}
		\babs{\hat d_j -d_j}< \frac{d_{\min}}{2},  \quad j=1,\cdots,k,
	\end{equation}
	and moreover
	\begin{equation}\label{equ:satblemultiproductlemma1equ5}
		\babs{\hat d_j -d_j}\leq \Big(\frac{2}{d_{\min}}\Big)^{k-1}\epsilon, \quad j=1,\cdots, k.
	\end{equation}
\end{lem}
\begin{proof}
See Appendix \ref{section:proofmultiproductstablem}. 
\end{proof}

\subsection{Lower-bound for the approximation problem (\ref{equ:vandermondeapprox})}\label{sec-vand-lowerbound}
In this section, we derive a lower-bound for the nonlinear approximation problem (\ref{equ:vandermondeapprox}). We first consider the special case when $v$ is a Vandermonde vector. 

\begin{thm}\label{thm:spaceapproxlowerbound0}
	Let $k\geq 1$ and $\hat d_1,\cdots, \hat d_{k}$ be $k$ distinct complex numbers with $|\hat d_j|\leq \hat d, 1\leq j\leq k$. Define $A := \big(\phi_{k}(\hat d_1), \cdots, \phi_k(\hat d_k)\big)$, where $\phi_k(\hat d_j)$'s are defined as in (\ref{equ:defineofphi}). Let $V$ be the $k$-dimensional space spanned by the column vectors of $A$, and let $V^{\perp}$ be the one-dimensional orthogonal complement of $V$ in $\mathbb C^{k+1}$. Let $P_{V^{\perp}}$ be the orthogonal projection onto $V^{\perp}$ in $\mathbb C^{k+1}$. Then we have 
	\[
	\min_{a\in \mathbb C^k}\btwonorm{Aa- \phi_{k}(x)} = \btwonorm{P_{V^{\perp}}(\phi_k(x))} =\babs{v^{*}\phi_{k}(x)}\geq \frac{1}{(1+\hat d)^k}\babs{\Pi_{j=1}^k (x- \hat  d_j)},
	\]
	where $v$ is a unit vector in $V^{\perp}$ and $v^{*}$ is its conjugate transpose.
\end{thm}
\begin{proof}
By Lemma \ref{lem:projectvolumeratiolem1}, it follows that
\[
\min_{a\in \mathbb C^k}\btwonorm{Aa- \phi_{k}(x)} = \sqrt{\frac{\text{det}(D^{*} D)}{\text{det}(A^{*}A)}},
\]
where $D=\big(\phi_{k}(\hat d_1),\cdots, \phi_{k}(\hat d_k), \phi_{k}(x)\big)$. Denote $\tilde{A} = \big(\phi_{k-1}(\hat d_1),\cdots, \phi_{k-1}(\hat d_k)\big)$. By (\ref{equ:projectiondisestimate1}), we have
\[
\sqrt{\frac{\text{det}(A^{*}A)}{\text{det}(\tilde{A}^{*}\tilde{A})}} \leq (1+\hat d)^k.
\]
Therefore, 
\[
\min_{a\in \mathbb C^k}\btwonorm{Aa- \phi_{k}(x)} \geq \frac{1}{(1+\hat d)^{k}}\sqrt{\frac{\text{det}(D^{*}D)}{\text{det}(\tilde{A}^{*}\tilde{A})}}.
\]
Note that $D$ and $\tilde{A}$ are square Vandermonde matrices. We can use the determinant formula to derive that 
\[
\min_{a\in \mathbb C^k}\btwonorm{Aa- \phi_{k}(x)} \geq \frac{1}{(1+\hat d)^{k}} \frac{|\Pi_{1\leq t<p\leq k} (\hat d_t - \hat d_p)\Pi_{q=1}^k(x-\hat d_q)|}{|\Pi_{1\leq t<p\leq k} (\hat d_t - \hat d_p)|} = \frac{1}{(1+\hat d)^k}|\Pi_{j=1}^k (x- \hat d_j)|.
\] 
This completes the proof of the theorem.
\end{proof}

\medskip

We now consider the approximation problem (\ref{equ:vandermondeapprox})
for the general case when $v$ is a linear combination of Vandermonde vectors. 

\begin{thm}\label{thm:spaceapproxlowerbound1}
	Let $k\geq 1$. Assume $(k+1)$ different complex numbers $d_j\in \mathbb C, j=1, \cdots, k+1$ with $|d_j|\leq d$ and $(k+1)$ $a_j\in \mathbb C$ with $|a_j|\geq m_{\min}$. Let $d_{\min}:=\min_{j\neq p}|d_j-d_p|$.  For $q\leq k$, let $\hat a(q)=(\hat a_1, \hat a_2, \cdots, \hat a_{q})^{\top}, a=(a_1, a_2, \cdots, a_{k+1})^{\top}$, and 
	\[
	\hat A(q) = \big(\phi_{2k}(\hat d_1),\ \cdots,\ \phi_{2k}(\hat d_q)\big),\  A = \big(\phi_{2k}(d_1),\ \cdots,\ \phi_{2k}(d_{k+1})\big),
	\]
	where $\phi_{2k}(z)$ is defined as in (\ref{equ:defineofphi}). Then
	\begin{align*}
		\min_{\hat a_p,\hat d_p\in \mathbb C, |\hat d_p|\leq \hat d, p=1,\cdots,q}||\hat A(q)\hat a(q)-Aa||_2\geq \frac{m_{\min}(d_{\min})^{2k}}{2^k(1+d)^{k}(1+\hat d)^k}.
	\end{align*}
\end{thm}
\begin{proof}
\textbf{Step 1}. Note that for $q<k$, we have 
\[\min_{\hat a_p, \hat d_p\in \mathbb C, |\hat d_p|\leq d, p=1,\cdots,q}||\hat A(q)\hat a(q)-Aa||_2\geq \min_{\hat a_p,\hat d_p\in \mathbb C, |\hat d_p|\leq d, p=1,\cdots,k}||\hat A(k)\hat a(k)-Aa||_2.\]
Hence we need only to consider the case when $q=k$. It then suffices to show that for any given $\hat d_j\in \mathbb C, |\hat d_j|\leq \hat d, j=1, \cdots, k$, the following holds
\begin{equation}\label{equ:spaceapproxlowerboundequ-4}
	\min_{\hat a_p\in \mathbb C, p=1,\cdots,k}||\hat A(k)\hat a(k)-Aa||_2\geq \frac{m_{\min}(d_{\min})^{2k}}{2^k(1+d)^{k}(1+\hat d)^k}.
\end{equation}
So we fix $\hat d_1,\cdots,\hat d_k$ in our subsequent argument.\\
\textbf{Step 2}. For $l=0, \cdots, k$, we define the following partial matrices 
\[
\hat A_{l}=
\left(\begin{array}{ccc}
	\hat d_1^l&\cdots&\hat d_k^l\\
	\hat d_1^{l+1}&\cdots &\hat d_k^{l+1}\\
	\vdots &\vdots &\vdots\\ 
	\hat d_1^{l+k}&\cdots&\hat d_k^{l+k}
\end{array}
\right),
\quad 
A_{l}=
\left(\begin{array}{ccc}
	(d_1)^l&\cdots&(d_{k+1})^l\\
	(d_1)^{l+1}&\cdots &(d_{k+1})^{l+1}\\
	\vdots &\vdots &\vdots\\ 
	(d_1)^{l+k}&\cdots& (d_{k+1})^{l+k}
\end{array}
\right).
\]
It is clear that for all $l$, 
\begin{equation}\label{equ:spaceapproxlowerboundequ-3}
	\min_{\hat a(k)\in \mathbb C^{k}}||\hat A(k)\hat a(k)-Aa||_2 \geq \min_{\hat a\in \mathbb C^k}||\hat A_l \hat a- A_l a||_2.
\end{equation}
\textbf{Step 3}. For each $l$, observe that $\hat A_l = \hat A_0 \text{diag}(\hat d_1^l,\ \cdots,\ \hat d_k^l), A_l=A_0\text{diag}(d_1^l, \cdots, d_{k+1}^l)$, and thus
\begin{equation}\label{equ:spaceapproxlowerboundequ-2}
	\min_{\hat a\in \mathbb C^{k}}||\hat A_l\hat a-A_la||_2\geq \min_{\hat \alpha_l\in \mathbb C^{k}}||A_0\hat \alpha_l-A_0\alpha_l||_2,
\end{equation}
where $\alpha_l=\left(a_1(d_1)^l,\cdots,a_{k+1}(d_{k+1})^l\right)^{\top}$. Let $V$ be the space spanned by the column vectors of $A_0$. Then the dimension of $V$ is $k$, and the dimension of $V^\perp$, the orthogonal complement of $V$ in $\mathbb C^{k+1}$, is one.  Let 
$P_{V^{\perp}}$ be the orthogonal projection onto $V^{\perp}$. Note that $||P_{V^{\perp}}u||_2=|v^{*}u|$ for $u\in \mathbb{R}^{k+1}$, where $v$ is a unit vector in $V^{\perp}$ and $v^{*}$ is its conjugate transpose. We have
\begin{align}\label{equ:spaceapproxlowerboundequ-1}
	\min_{\hat \alpha_l\in \mathbb C^k}||\hat A_0\hat \alpha_l-A_0\alpha_l||_2=||P_{V^{\perp}}(A_0\alpha_l)||_2=|v^{*}A_0\alpha_l |=\Big|\sum_{j=1}^{k+1}a_j(d_j)^l v^{*} \phi_{k}(d_j)\Big| = |\beta_l|,
\end{align} 
where
\[
\beta_l = \sum_{j=1}^{k+1}a_j(d_j)^l v^{*} \phi_{k}(d_j), \quad \text{for $l=0, 1, \cdots, k.$}
\] 
\textbf{Step 4}. Denote $\beta = (\beta_0, \cdots, \beta_k)^{\top}$. We have $B\hat\eta=\beta$, where
\[B=\left(\begin{array}{cccc}
	a_1&a_2&\cdots&a_{k+1}\\
	a_1d_1&a_2d_2&\cdots&a_{k+1}d_{k+1}\\
	\vdots&\vdots&\vdots&\vdots\\
	a_1(d_1)^{k}&a_2(d_2)^{k}&\cdots&a_{k+1}(d_{k+1})^k
\end{array}\right),\quad \hat \eta=
\left(\begin{array}{c}
	v^{*} \phi_{k}(d_1)\\
	v^{*} \phi_{k}(d_2)\\
	\vdots\\
	v^{*} \phi_{k}(d_{k+1})
\end{array}\right).
\]
Corollary \ref{lem:norminversevandermonde1} yields
\begin{align*}
	||\hat \eta||_{\infty}=||B^{-1}\beta||_{\infty}\leq ||B^{-1}||_{\infty}||\beta||_{\infty}\leq \frac{(1+d)^{k}}{m_{\min}(d_{\min})^{k}}||\beta||_{\infty}.
\end{align*}
On the other hand, applying Theorem \ref{thm:spaceapproxlowerbound0} to each term $|v^{*} \phi_{k}(d_j)|$, $j=1, 2, \cdots k+1$,  we have 
\[
||\hat \eta||_{\infty}\geq \frac{1}{(1+\hat d)^k} ||\eta_{k+1, k}(d_1, \cdots, d_{k+1}, \hat d_1, \cdots, \hat d_{k} )||_{\infty},
\]
where $\eta_{k+1, k}(\cdots)$ is defined as in (\ref{equ:defineofeta}). Combining this inequality with Lemma \ref{lem:multiproductlowerbound0}, we get 
\[
||\hat \eta||_{\infty}\geq \frac{(d_{\min})^{k}}{2^k(1+\hat d)^k}.
\]
Then it follows that 
\[
||\beta||_{\infty} \geq \frac{m_{\min}(d_{\min})^{2k}}{2^k(1+d)^{k}(1+\hat d)^{k}}.
\]
Therefore, recalling (\ref{equ:spaceapproxlowerboundequ-3})--(\ref{equ:spaceapproxlowerboundequ-1}), we arrive at
\[
\min_{\hat a(k)\in \mathbb C^{k}}||\hat A(k)\hat a(k)-Aa||_2\geq \max_{0\leq l\leq k} \min_{\hat a\in \mathbb C^k}||\hat A_l\hat a-A_la||_2= \max_{0\leq l\leq k} |\beta_l|=||\beta||_{\infty}\geq \frac{m_{\min}(d_{\min})^{2k}}{2^k(1+d)^{k}(1+\hat d)^k}.
\]
This proves (\ref{equ:spaceapproxlowerboundequ-4}) and hence the theorem.
\end{proof}
\subsection{Stability of the approximation problem (\ref{equ:vandermondeapprox})}\label{sec-vand-stability}
In the section we present a stability result for the approximation problem (\ref{equ:vandermondeapprox}).

\begin{thm}\label{spaceapproxlowerbound3}
	Let $k\geq 1$. Assume $k$ different complex numbers $d_j\in \mathbb C, j=1, \cdots, k$ with $|d_j|\leq d$ and $k$ $a_j\in \mathbb C$ with $|a_j|\geq m_{\min}$. Let $d_{\min}:=\min_{p\neq q}|d_p-d_q|$. Assume that $\hat d_j\in \mathbb C, j=1, \cdots,k$ with $|\hat d_j|\leq d$ satisfy
	\[
	||\hat A\hat a-Aa||_2< \sigma, 
	\]
	where $\hat a=(\hat a_1,\cdots, \hat a_k)^{\top}$, $a = (a_1,\cdots, a_k)^{\top}$, and 
	\[
	\hat A = \big(\phi_{2k-1}(\hat d_1),\ \cdots,\ \phi_{2k-1}(\hat d_k)\big),\  A = \big(\phi_{2k-1}(d_1),\ \cdots,\ \phi_{2k-1}(d_{k})\big).
	\]
	Then 
	\[
	\Big|\Big|\eta_{k,k}(d_1,\cdots, d_k, \hat d_1, \cdots, \hat d_k)\Big|\Big|_{\infty}<\frac{(1+d)^{2k-1}}{ d_{\min}^{k-1}}\frac{\sigma}{m_{\min}}.
	\]
\end{thm}
\begin{proof}
Since $||\hat A\hat a-Aa||_2<\sigma$, we have 
\begin{equation*}\label{spaceapproxlowerbound3equ0}
	\min_{\hat \alpha \in \mathbb C^k}||\hat A\hat \alpha-Aa||_2<\sigma,
\end{equation*}
and hence
\begin{equation}\label{spaceapproxlowerbound3equ1}
	\max_{0\leq l\leq k-1}\min_{\hat \alpha \in \mathbb C^k}||\hat A_l \hat \alpha -A_la||_2\leq \min_{\hat \alpha \in \mathbb C^k}||\hat A\hat \alpha-Aa||_2<\sigma,
\end{equation}
where 
\[
\hat A_{l}=
\left(\begin{array}{ccc}
	\hat d_1^l&\cdots&\hat d_k^l\\
	\hat d_1^{l+1}&\cdots &\hat d_k^{l+1}\\
	\vdots &\vdots &\vdots\\ 
	\hat d_1^{l+k}&\cdots&\hat d_k^{l+k}
\end{array}
\right),
\quad 
A_{l}=
\left(\begin{array}{ccc}
	d_1^l&\cdots&d_{k}^l\\
	d_1^{l+1}&\cdots &d_{k}^{l+1}\\
	\vdots &\vdots &\vdots\\ 
	d_1^{l+k}&\cdots& d_{k}^{l+k}
\end{array}
\right).
\]
For each $l$, from the decomposition $\hat A_l = \hat A_0 \text{diag}(\hat d_1^l,\ \cdots,\ \hat d_k^l), A_l=A_0\text{diag}((d_1)^l, \cdots, (d_{k})^l)$, we get
\begin{equation}\label{spaceapproxlowerbound3equ2}
	\min_{\hat \alpha \in \mathbb C^{k}}||\hat A_l\hat \alpha-A_la||_2\geq \min_{\hat \alpha_l\in \mathbb C^{k}}||\hat A_0\hat \alpha_l-A_0\alpha_l||_2,
\end{equation}
where $\alpha_l=(a_1(d_1)^l,\cdots,a_{k}(d_{k})^l)^{\top}$. Let $V$ be the space spanned by the column vectors of $\hat A_0$. Then the dimension of $V$ is $k$, and $V^\perp$, the orthogonal complement of $V$ in $\mathbb C^{k+1}$ is of dimension one. We let $v$ be a unit vector in $V^{\perp}$ and let 
$P_{V^{\perp}}$ be the orthogonal projection onto $V^{\perp}$. Similarly to (\ref{equ:spaceapproxlowerboundequ-1}), we have
\begin{align}\label{spaceapproxlowerbound3equ3}
	\min_{\hat \alpha_l\in \mathbb C^k}||\hat A_0\hat \alpha_l-A_0\alpha_l||_2=||P_{V^{\perp}}(A_0\alpha_l)||_2=|v^{*}A_0 \alpha_l|=\Big|\sum_{j=1}^{k}a_j(d_j)^lv^{*}\phi_{k}(d_j)\Big|=|\beta_l|,
\end{align} 
where $\beta_l = \sum_{j=1}^{k}a_j(d_j)^lv^{*}\phi_{k}(d_j)$.  Let $\beta = (\beta_0,\cdots, \beta_{k-1})^{\top}$. Moreover, similarly to Step 4 in the proof of Theorem \ref{thm:spaceapproxlowerbound1}, we have 
\[
\Big|\Big|\eta_{k,k}(d_1,\cdots, d_k, \hat d_1, \cdots, \hat d_k)\Big|\Big|_{\infty} \leq \frac{(1+d)^{2k-1}}{m_{\min}(d_{\min})^{k-1}} ||\beta||_{\infty}.
\]
On the other hand,  (\ref{spaceapproxlowerbound3equ1})--(\ref{spaceapproxlowerbound3equ3}) indicate that $||\beta||_{\infty}<\sigma$. Hence, we obtain that
\[
\Big|\Big|\eta_{k,k}(d_1,\cdots, d_k,\hat  d_1, \cdots,\hat  d_k)\Big|\Big|_{\infty} \leq \frac{(1+d)^{2k-1}}{(d_{\min})^{k-1}}\frac{\sigma}{m_{\min}}.
\]
This completes the proof. 
\end{proof}




\section{Conclusions and future works}
In this paper, we have improved the estimates of resolution limits in two-dimensional super-resolution problems. We also theoretically demonstrate the optimal performance of a sparsity-promoting algorithm. Leveraging the new techniques in the proof, we have proposed a coordinate-combination-based model order detection algorithm and a coordinate-combination-based MUSIC algorithm for DOA estimation in two dimensions. The superiority of the introduced algorithms were demonstrated both theoretically or numerically. 

Our work is also a start of many new topics. Firstly, one could extend the techniques to three- and $k$-dimensional spaces to improve the resolution estimates in higher dimensional super-resolution problems. Secondly, the idea of coordinate-combination could inspire new algorithms for two-dimensional DOA estimations in the case of multiple snapshots. These works will be presented in a near future.

%

\appendix 	

\section{Proof of Lemma \ref{lem:multiproductstability1}}\label{section:proofmultiproductstablem}
\begin{proof}
\textbf{Step 1.} We claim that 
for each $\hat d_p, 1\leq p\leq k$, there exists one $d_j$ such that $|\hat d_p-d_j|<\frac{d_{\min}}{2}$.
By contradiction, suppose that there exists $p_0$ such that $|d_j - \hat d_{p_0}|\geq \frac{d_{\min}}{2}$ for all $1\leq j\leq k$. Observe that 
\begin{align*}
&\eta_{k,k}(d_1,\cdots,d_k, \hat d_1,\cdots, \hat d_k)\\
=&\text{diag}\left(|
d_1- \hat d_{p_0}|,\cdots,|d_{k}-\hat d_{p_0}|\right)\eta_{k,k-1}(d_1,\cdots,d_{k},\hat  d_1,\cdots,\hat d_{p_0-1}, \hat d_{p_0+1},\cdots,\hat d_k).
\end{align*}
We write $$\eta_{k,k}=\eta_{k,k}(d_1,\cdots,d_k, \hat d_1,\cdots, \hat d_k) \quad \mbox{and} \quad \eta_{k,k-1}=\eta_{k,k-1}(d_1,\cdots,d_{k}, \hat d_1,\cdots,\hat d_{p_0-1}, \hat d_{p_0+1},\cdots,\hat d_k).$$ Using Lemma \ref{lem:multiproductlowerbound0}, we have
\[
||\eta_{k,k}||_{\infty}\geq \frac{d_{\min}}{2}||\eta_{k, k-1}||_{\infty} \geq \Big(\frac{d_{\min}}{2}\Big)^{k}\geq \epsilon,
\]	
where we have used (\ref{equ:satblemultiproductlemma1equ2}) in the last inequality above. This contradicts (\ref{equ:satblemultiproductlemma1equ1}) and hence proves our claim.\\
\textbf{Step 2.} We claim that for each $d_j, 1\leq j\leq k$, there exists one and only one $\hat d_p$ such that $$|d_j - \hat d_p|< \frac{d_{\min}}{2}.$$ It suffices to show that for each $d_j, 1\leq j\leq k$, there is only one $\hat d_p$ such that $|d_j-\hat d_p|<\frac{d_{\min}}{2}$. By contradiction, suppose that there exist $p_1,p_2,$ and $j_0$ such that $|d_{j_0}-\hat d_{p_1}|<\frac{d_{\min}}{2}, |d_{j_0}-\hat d_{p_2}|<\frac{d_{\min}}{2}$. Then for all $j \neq j_0$, we have 
\begin{equation}\label{equ:satblemultiproductlemma1equ3}
	\Big|(d_j- \hat d_{p_1}) (d_j-\hat d_{p_2})\Big|\geq \frac{(d_{\min})^2}{4}.
\end{equation}
Similarly to the argument in Step 1, we separate the factors involving $\hat d_{p_1}, \hat d_{p_2}, d_{j_0}$ from $\eta_{k,k}$ and  consider
\[
\eta_{k-1,k-2}=	\eta_{k-1,k-2}(d_1,\cdots,d_{j_0-1},d_{j_0+1},\cdots,d_k, \hat d_1,\cdots, \hat d_{p_1-1}, \hat d_{p_1+1}, \cdots, \hat d_{p_2-1}, \hat d_{p_2+1},\cdots, \hat d_k).
\]
Note that the components of $\eta_{k-1,k-2}$ differ from those of $\eta_{k,k}$ only by the factors $|(d_j-\hat d_{p_1})(d_{j}-\hat d_{p_2})|$ for $j=1,\cdots,j_0-1,j_0+1,\cdots,k$. We can show that
$$	
||\eta_{k,k}||_{\infty}\geq \frac{(d_{\min})^2}{4} ||\eta_{k-1,k-2}||_{\infty} \geq \epsilon, 
$$
where we have used Lemma \ref{lem:multiproductlowerbound0} and (\ref{equ:satblemultiproductlemma1equ2}) for establishing the last inequality above. This contradicts (\ref{equ:satblemultiproductlemma1equ1}) and hence proves our claim. \\
\textbf{Step 3.} By the result in Step 2,  we can reorder $\hat d_j$'s to get
\[
|\hat d_j -d_j|< \frac{d_{\min}}{2}, \quad j=1,\cdots,k.
\]
We now prove (\ref{equ:satblemultiproductlemma1equ5}). It is clear that $|\hat d_{p} - d_{j}|> \frac{d_{\min}}{2}, p\neq j$. Thus
\begin{equation} \label{eq-222}
	|(d_j-\hat d_1)\cdots(d_j- \hat d_k)|> |d_{j}-d_j|(\frac{d_{\min}}{2})^{k-1}, \quad j=1, 2, \cdots, k. 
\end{equation}  
Further, we get 
\[
|d_{j} -\hat d_j|< \Big(\frac{2}{d_{\min}}\Big)^{k-1} ||\eta_{k,k}||_{\infty}\leq \Big(\frac{2}{d_{\min}}\Big)^{k-1} \epsilon, \quad j=1, 2, \cdots, k.
\]
This completes the proof of the lemma. 
\end{proof}
\bibliographystyle{plain}
\bibliography{references_final}	

\end{document}